\documentclass[sigconf,nonacm]{acmart}
\acmSubmissionID{}
\AtBeginDocument{%
  \providecommand\BibTeX{{%
    Bib\TeX}}}
\setcopyright{none}
\copyrightyear{2026}
\acmYear{2026}
\acmDOI{}

\acmISBN{}

\newif\ifshowappendix
\showappendixtrue

\renewcommand\footnotetextcopyrightpermission[1]{}
\usepackage{amsmath}
\usepackage{amsthm}
\usepackage{pifont}
\usepackage{xspace}
\usepackage{threeparttable}
\usepackage{color}
\usepackage{colortbl}
\usepackage{url}
\usepackage{subcaption}
\usepackage{paralist}
\usepackage{wrapfig}
\usepackage{multirow}
\usepackage{listings}
\usepackage{comment}
\usepackage{booktabs}
\usepackage{makecell}
\usepackage{setspace}
\usepackage{boxedminipage}
\usepackage{graphicx}
\usepackage{subfloat} 
\usepackage{minibox}
\usepackage{pifont}
\usepackage{color}
\usepackage{xcolor}
\usepackage[table]{xcolor}
\usepackage{subfiles}
\usepackage[utf8]{inputenc}
\usepackage[english]{babel}
\usepackage{ulem}
\usepackage{adjustbox}
\definecolor{blue}{RGB}{10, 10, 200}
\usepackage{amsmath}
\usepackage{titlesec}
\usepackage{placeins}
\usepackage{enumitem}
\usepackage{pagesel}
\usepackage{thmtools}
\usepackage{thm-restate}
\usepackage{atbegshi}

\def\BibTeX{{\rm B\kern-.05em{\sc i\kern-.025em b}\kern-.08em
    T\kern-.1667em\lower.7ex\hbox{E}\kern-.125emX}}

\newcommand{\bheading}[1]{{\vspace{4pt}\noindent{\textbf{#1}}}}
\newcommand{\iheading}[1]{{\vspace{2pt}\noindent{\textit{#1}}}}
\newcolumntype{?}{!{\vrule width 1pt}}

\newcounter{note}[section]
\renewcommand{\thenote}{\thesection.\arabic{note}}

\ifx\submission\undefined
    \newcommand{\Sadoghi}[1]{{\color{cyan}{\textbf{Sadoghi:} #1}}}
    \newcommand{\niu}[1]{\refstepcounter{note}{\bf\textcolor{red}{$\ll$JN~\thenote: {\sf #1}$\gg$}}}
    \newcommand{\hz}[1]{\refstepcounter{note}{\bf\textcolor{purple}{$\ll$HZ~\thenote: {\sf #1}$\gg$}}}
    \newcommand{\shaokang}[1]{\refstepcounter{note}{\bf\textcolor{blue}{$\ll$SK~\thenote: {\sf #1}$\gg$}}}
    \newcommand{\Dakai}[1]
    {{\color{orange}{\textbf{Dakai:} #1}}}
    \newcommand{\new}[1]{#1}
\else
    \newcommand{\Sadoghi}[1]{}
    \newcommand{\niu}[1]{}
    \newcommand{\hz}[1]{}
    \newcommand{\shaokang}[1]{}
    \newcommand{\Dakai}[1]{}

\fi

\newcommand{\apxref}[1]{%
  Appendix~\ref{#1}~\cite{fides_appendix}%
}

\newcommand{\apxlmref}[2]{%
  \ref{#1} in %
  \ifshowappendix
    Appendix~\ref{#2}%
  \else
    the Appendix~\cite{fides_appendix}%
  \fi
}

\newcommand{\secref}[1]{\mbox{Sec.~\ref{#1}}\xspace}

\newcommand{\figref}[1]{\mbox{Fig.~\ref{#1}}}

\newcommand{\ignore}[1]{}

\newcommand{\ie}{\textit{i.e.}\xspace}
\newcommand{\eg}{\textit{e.g.}\xspace}

\newcommand{\sysname}{\textsc{{F}ides}\xspace}

\newcommand{\blackding}[1]{\ding{\numexpr181+#1\relax}}

\newcounter{packednmbr}
\newenvironment{packedenumerate}{
\begin{list}{\thepackednmbr.}{\usecounter{packednmbr}
\setlength{\itemsep}{0pt}
\addtolength{\labelwidth}{4pt}
\setlength{\leftmargin}{12pt}
\setlength{\listparindent}{\parindent}
\setlength{\parsep}{3pt}
\setlength{\topsep}{3pt}}}{\end{list}}

\newenvironment{packeditemize}{
\begin{list}{$\bullet$}{
\setlength{\labelwidth}{0pt}
\setlength{\itemsep}{2pt}
\setlength{\leftmargin}{\labelwidth}
\addtolength{\leftmargin}{\labelsep}
\setlength{\parindent}{0pt}
\setlength{\listparindent}{\parindent}
\setlength{\parsep}{1pt}
\setlength{\topsep}{1pt}}}{\end{list}}

\newtheorem{theorem}{Theorem}
\newtheorem{lemma}{Lemma}

\newtheorem{claim}{Claim}

\usepackage{tikz}
\usetikzlibrary{arrows.meta, shapes.geometric, positioning,calc}

\newcommand{\Name}[1]{\textnormal{\textsc{#1}}}

\usepackage[noend]{algorithmic}
\newcommand{\GETS}{:=}
\newenvironment{myprotocol}{
    \hrule
    \smallskip
    \algsetup{linenosize=\small}
    \begin{algorithmic}[1]
        
        \newcommand{\SPACE}{\item[]}
        \newcommand{\TITLE}[2]{\item[] \textbf{\underline{##1}} (##2) \textbf{:}\\[2pt]}
        \makeatletter
            \newcommand{\EVENT}[2][]{\STATE \textbf{event} ##2 \textbf{do}%
                \ifthenelse{\equal{##1}{}}{}{\ \algorithmiccomment{##1}}%
                \begin{ALC@g}}
            \newcommand{\ENDEVENT}{\end{ALC@g}}
        \makeatother
        
        \makeatletter
            \newcommand{\FUNCTION}[2]{\STATE \textbf{function} \Name{##1}(##2) \textbf{do}\begin{ALC@g}}
            \newcommand{\ENDFUNCTION}{\end{ALC@g}}
        \renewcommand{\COMMENT}[1]{\hfill $\triangleright$ ##1}
        \renewcommand{\algorithmiccomment}[1]{\hfill $\triangleright$ ##1}
        \makeatother
    
}{
    \end{algorithmic}%
    \hrule
}

\newtheorem{theoremapx}{Theorem}

\titlespacing*{\section}{0pt}{1.4ex}{1.4ex}
\titlespacing*{\subsection}{0pt}{1.2ex}{1.2ex}
\titlespacing*{\subsubsection}{0pt}{0.8ex}{0.8ex}

\begin{document}

\title{\sysname: {S}ecure and {S}calable {A}synchronous {DAG} {C}onsensus \\ via  {T}rusted {C}omponents}

\author{Shaokang Xie$^1$, Dakai Kang$^1$, Hanzheng Lyu$^2$, Jianyu Niu$^3$, Mohammad Sadoghi$^1$}
\affiliation{%
  \institution{
  $^1$Exploratory Systems Lab, Department of Computer Science, University of California, Davis \\ 
  $^2$School of Engineering, University of British Columbia (Okanagan campus) \\
  $^3$Department of Computer Science, City University of Hong Kong}
}

\begin{abstract}
DAG-based BFT consensus has attracted growing interest in distributed data management systems for consistent replication in untrusted settings due to its high throughput and resilience to asynchrony.
However, existing protocols still suffer from high communication overhead and long commit latency. In parallel, introducing minimal hardware trust has proven effective in reducing the complexity of BFT consensus.

Inspired by these works, we present \sysname, an asynchronous DAG-based BFT consensus protocol that, to our knowledge, is among the first to leverage TEEs to enhance both scalability and efficiency.
\sysname tolerates a minority of Byzantine replicas and achieves $\mathcal{O}(\kappa n^2 + n^3)$ metadata communication complexity 
through a customized TEE-assisted Reliable Broadcast (\textit{T-RBC}) primitive with linear communication complexity in one-step broadcast. 
Building on T-RBC, \sysname redefines the DAG construction rules by reducing the reference requirement from $2f{+}1$ to $f{+}1$ between consecutive vertices. 
This new structure weakens DAG connectivity and invalidates traditional commit rules, so we formally abstract the problem and derive new theoretical bounds on liveness.
We further propose a \textit{four-round commit rule} that achieves the theoretically minimal commit latency.
In addition, we design two additional primitives, \textit{T-RoundCert} and \textit{T-Coin}, to efficiently certify DAG references and replace the costly cryptographic common coin used in prior protocols.
Comprehensive evaluations on geo-distributed and local testbeds show that \sysname substantially outperforms state-of-the-art protocols, including Tusk, Bullshark, Mysticeti, \new{Shoal++}, RCC, Damysus, Achilles and HybridSet, achieving lower latency and higher throughput while preserving strong safety and liveness guarantees. 
\end{abstract}

\maketitle


\nocite{fides_appendix, sourcecode,intelRA,sgxsealing,amdRA,intel_software_security_guidance,sgxsimulate,resdb,openenclave2022, aptos, sui, osmosis, jupiter}

\section{Introduction} \label{sec:intro}
Modern distributed data management systems rely on replication to keep data consistent and available despite failures. 
As deployments scale to unreliable, geo-distributed, and even untrusted environments, Byzantine fault tolerance (BFT) consensus has gained significant interest as a backbone for fault-tolerant replicated databases.
Notable examples include the rapid rise of blockchains~\cite{tendermint, aptos, sui} and Web3 applications~\cite{osmosis, jupiter}, 
which require BFT to order transactions in decentralized and Byzantine networks. 
However, geo-distributed deployments with fluctuating network conditions~\cite{Chandra1996, Clement2009, avarikioti2020fnf, 2023fever} pose new challenges for classic BFT protocols, such as PBFT~\cite{pbft1999}, which are typically designed for partially synchronous networks~\cite{consensusinpartialsync}. 
Prior studies have shown that these protocols may fail to guarantee liveness (\ie, the sequence of agreed transactions is ever-growing) in fully asynchronous environments~\cite{consensusinpartialsync, pbft1999, honeybadger, fairness}.


To remove the reliance on network synchrony, a series of fully asynchronous protocols~\cite{honeybadger, fairness, Dumbo-NG, Dumbo-MVBA, DAGRider, narwhal, Bullshark, gradeddag, jovanovic2024mahi} have been proposed to ensure consensus under arbitrary message delays.
Among them, asynchronous \textit{Directed Acyclic Graph} (DAG)-based protocols~\cite{DAGRider, Bullshark, narwhal, gradeddag, jovanovic2024mahi} have gained particular prominence due to their efficiency, robustness, and conceptual simplicity.
They form a DAG of proposals from all replicas, avoid single-leader bottlenecks, commit transactions directly from the DAG without extra communication, and preserve asynchronous liveness with optimal amortized communication complexity. 
Representative examples include Tusk~\cite{narwhal} and Bullshark~\cite{Bullshark}.

Asynchronous DAG-based protocols require $n{=}3f{+}1$ replicas to tolerate up to $f$ Byzantine ones (\ie, behaving arbitrarily). They proceed in waves, each comprising multiple rounds.
In each round, all replicas use reliable broadcast (RBC) to disseminate proposals (\ie, vertices), each carrying pending transactions and references to at least $2f{+}1$ proposals from the previous round, thereby forming a DAG.
At the end of each wave, replicas invoke a cryptographic common coin to collectively select a leader vertex from the first round of that wave, and decide to commit the leader vertex and its referencing proposals according to the commit rules. 
This enables replicas to produce a consistently ordered sequence of transactions. 

Despite their promising features, DAG-based BFT protocols still struggle to scale in large-scale decentralized applications. \textit{First}, the \textit{large system size} requirement of $n = 3f{+}1$ replicas causes performance (\eg, throughput) to degrade sharply once the number of replicas reaches several tens (see \secref{sec:evaluation}). 
\textit{Second}, these protocols incur \textit{high communication complexity}: each round runs $n$ instances of reliable broadcast (RBC), each with $\Theta(n^2)$ communication cost, and requires validating $\mathcal{O}(n^2)$ references per round (up to $n$ vertices)~\cite{DAGRider, narwhal, BLS2001}.
\textit{Third}, they experience \textit{long commit latency}, especially under adversarial asynchronous networks. 
Table~\ref{tab:dag-comparison} shows that asynchronous DAG protocols guarantee commitment but still require either many message steps or larger committees, while partially synchronous DAG protocols such as Shoal++, Sailfish, and Mysticeti fail to ensure progress under asynchrony.

In this paper, we introduce \sysname, an asynchronous DAG-based BFT consensus protocol that, to our knowledge, is among the first to leverage Trusted Execution Environments (TEEs) to enhance both scalability and efficiency. 
However, using TEEs in DAG protocols is \textit{not plug-and-play.} 
The key challenge is to \textit{identify the specific bottlenecks in DAG-based consensus} and address them with trusted components built atop TEEs by following the minimized Trusted Computing Base (TCB) principle~\cite{trInc, hybster, Damysus, FlexiTrust, OneShot, Achilles}.  
To this end, we employ Monotonic Counter (MC), a trusted component to prevent message equivocation by binding each message to a unique counter value. 
With it, we design a \textit{TEE-assisted Reliable Broadcast (T-RBC)} primitive with linear communication complexity in one-step broadcast.
By incorporating T-RBC, \sysname can tolerate a minority of Byzantine replicas and achieve $\mathcal{O}(\kappa n^2 + n^3)$ communication complexity. 

T-RBC also enables us to relax the DAG construction rule: each vertex now needs only $f{+}1$ references to the previous round, instead of $2f{+}1$.
However, this relaxation creates a \textit{new liveness pitfall} that is specific to TEE-assisted DAG protocols and is absent from prior TEE-assisted (non-DAG) BFT designs~\cite{minbft, trInc, hybster, OneShot, Achilles, FlexiTrust, TrustedHardwareAssisted, dqbft}. 
In Tusk’s wave-based commitment~\cite{narwhal}, progress implicitly relies on sufficient overlap between the reference sets of consecutive rounds.
With only $f{+}1$ references, the overlap between two such sets can drop to one honest replica in the worst case (rather than at least $f{+}1$ under the original $2f{+}1$ rule).
An adversary that controls message scheduling can then indefinitely delay that replica’s messages, suppressing the only bridging references and causing the protocol to lose liveness under adversarial asynchrony.


To resolve this pitfall, we abstract the wave-level commit conditions as a $k$-iteration common core problem (\secref{sec:bound}), building on insights from prior DAG protocols~\cite{DAGRider, narwhal, Bullshark}. The parameter $k$ is the number of rounds per wave. 
Using this abstraction, we derive that the minimum $k=3$ suffices for liveness, while setting $k=4$ is necessary to achieve constant-round commit latency.
The gap stems from the fact that, with $k=3$, commits can require $\Theta(f)$ additional rounds in the worst case; in contrast, increasing to $k=4$ yields constant-round commits and achieves the optimal expected commit latency within our model.

We propose the Round Advancement Certifier (RAC), a trusted component that validates a vertex's collected references to the previous round and issues a compact cryptographic certificate when the references are valid. Building on RAC, we design \textit{TEE-assisted Round Certifier (T-RoundCert)} so that replicas can quickly verify references and advance to the next round.  
Besides, we propose \textit{TEE-assisted Common Coin (T-Coin)} that utilizes TEEs' confidential property to provide a lightweight Common Coin, replacing costly threshold-cryptographic constructions. 

We built end-to-end prototypes of \sysname atop the open-source Apache ResilientDB platform~\cite{resdb, resdbpaper} and used the Open Enclave SDK~\cite{openenclave2022} to develop trusted components on \textit{Intel SGX}~\cite{intelsgx}.
We conducted extensive experiments in the public cloud to evaluate and compare \sysname with several representative protocols, including Tusk~\cite{narwhal}, Bullshark~\cite{Bullshark}, Mysticeti~\cite{babel2024mysticeti}, \new{Shoal++~\cite{shoal++}}, RCC~\cite{gupta2021rcc}, Damysus~\cite{Damysus}, Achilles~\cite{Achilles}, and HybridSet~\cite{TrustedHardwareAssisted}.
Evaluation results in geo‑distributed and local environments show that \sysname achieves peak throughput of 400k tps and 810k tps, respectively, the highest among all compared protocols. We also examine its runtime overhead through fault‑injection tests and a detailed stage‑by‑stage and component-based performance breakdown.

\bheading{Contributions.} We propose \sysname, to our knowledge, one of the first asynchronous TEE-assisted DAG-based consensus protocols: 
\begin{packeditemize}
    \item \textbf{TEE-assisted DAG with $n{=}2f{+}1$ replicas and $\mathcal{O}(\kappa n + n^2)$ bits per proposal.} 
    We reduce the replica requirement of asynchronous DAG protocols from $n{=}3f{+}1$ to $n{=}2f{+}1$ while tolerating $f$ Byzantine faults, enabled by \textit{T-RBC}. Additionally, \textit{T-RBC} allows \sysname to reduce dissemination and referencing overhead, improving message complexity to $\mathcal{O}(\kappa n + n^2)$ bits per proposal.  


    \item \textbf{Reduced-quorum liveness and commitment.} We identify a liveness pitfall introduced by relaxing cross-round references and develop a four-round commit rule, achieving constant (expected) commit latency under adversarial asynchrony.  
    

    \item \textbf{Lower computation overhead.} \sysname avoids crypto-heavy components (notably threshold-cryptographic common coins) via \textit{T-RoundCert} (powered by \textit{RAC}) and a lightweight TEE-backed common coin (\textit{T-Coin} with \textit{RNG}). 
    
    \item \textbf{Implementation and evaluation.} We implement \sysname and evaluate it against state-of-the-art protocols, attributing costs and quantifying the performance gains of each trusted component. 
\end{packeditemize}

\section{Background and Motivation} \label{sec:motivation}

\subsection{Trusted Execution Environment}
Trusted Execution Environments (TEEs) are hardware-enforced, isolated execution environments that protect code and data even from a malicious OS or hypervisor. 
For example, in Intel SGX, the isolated environments are known as enclaves, whose memory is transparently encrypted and isolated.
Besides, there are two key mechanisms of TEEs: attestation, which enables TEEs to prove remotely that a particular enclave with a known code identity is running securely, and sealing, which encrypts enclave state for persistent storage tied to enclave identity~\cite{intelRA, sgxsealing, intelsgxpaper}.

\subsection{Asynchronous DAG-Based Consensus}
Asynchronous DAG-based consensus protocols~\cite{DAGRider, narwhal, Bullshark} operate round-by-round. 
In each round, every replica proposes a block containing a batch of transactions, and each block references at least $2{f}+1$ blocks (\ie, $\Theta(n)$) from the previous round.
These blocks act as \textit{vertices}, and the reference links between them serve as \textit{edges}, collectively forming a DAG. 
As a result, a replica has to verify $\Theta(n^2)$ references given at most $n$ proposals in a round during the DAG construction of one round. 
Each reference is a $\kappa$-bit hash~\cite{DAGRider, narwhal, BLS2001}.
After vertex validation, the proposed blocks then go through \textit{reliable broadcast (RBC)} to be delivered by replicas.
The agreement property of RBC ensures at most one vertex from a replica is produced for each round (\apxref{app:rbc}). 
The above process is also referred to as the dissemination phase.

Rounds are grouped into consecutive waves, each wave comprising a specific number of rounds. The rounds in a wave collectively form a \textit{common core} problem, determining how many vertices from the first round can be safely committed~\cite{DAGRider, narwhal, Bullshark}.
At the end of each wave, a vertex is randomly elected as the \textit{leader vertex} using a \textit{global common coin} (\apxref{app:commoncoin}), which produces the same random value across all replicas. 
If the leader vertex is in the common core, replicas can commit the vertex and its referenced vertices (\ie, transactions) in a consistent order. The commitment requires no additional communication overhead to achieve transaction ordering.

\subsection{Dissecting DAG Consensus} \label{subsec:dissect}
We dissect the performance of Tusk, a representative asynchronous DAG protocol that is extended by many later DAG protocols~\cite{narwhal, Bullshark, shoal}. 
We evaluate throughput and latency of Tusk under varying numbers of nodes in a LAN setting, as shown in \figref{fig:motivation}.
The experimental results reveal the following findings. 

\begin{packeditemize}
\item \textbf{Large system size leads to poor scalability.} Throughput decreases and latency increases as the system scales, reflecting the potential inherent cost of a large quorum.

\item \textbf{RBC dominates latency.}
RBC accounts for one of the largest parts of end-to-end latency (15\% in LAN and more than 50\% in WAN), making it a primary scalability bottleneck.

\item \textbf{Referencing overhead at scale.}
Encoding $\Theta(n)$ references to previous-round vertices inflates per-vertex metadata by $\Theta(n)$, slowing propagation and validation at large $n$.


\item \textbf{Common coin primitives are costly and fragile.}
Tusk's threshold-signature-based coin adds about 5--10\% to end-to-end latency.
Worse, recent work on cross-block front-running attacks reveals the coin's potential vulnerabilities~\cite{bluefish}. 


\end{packeditemize}

\begin{figure}[t]
    \vspace{-4mm}
    \centering
    \includegraphics[width=0.99\linewidth]{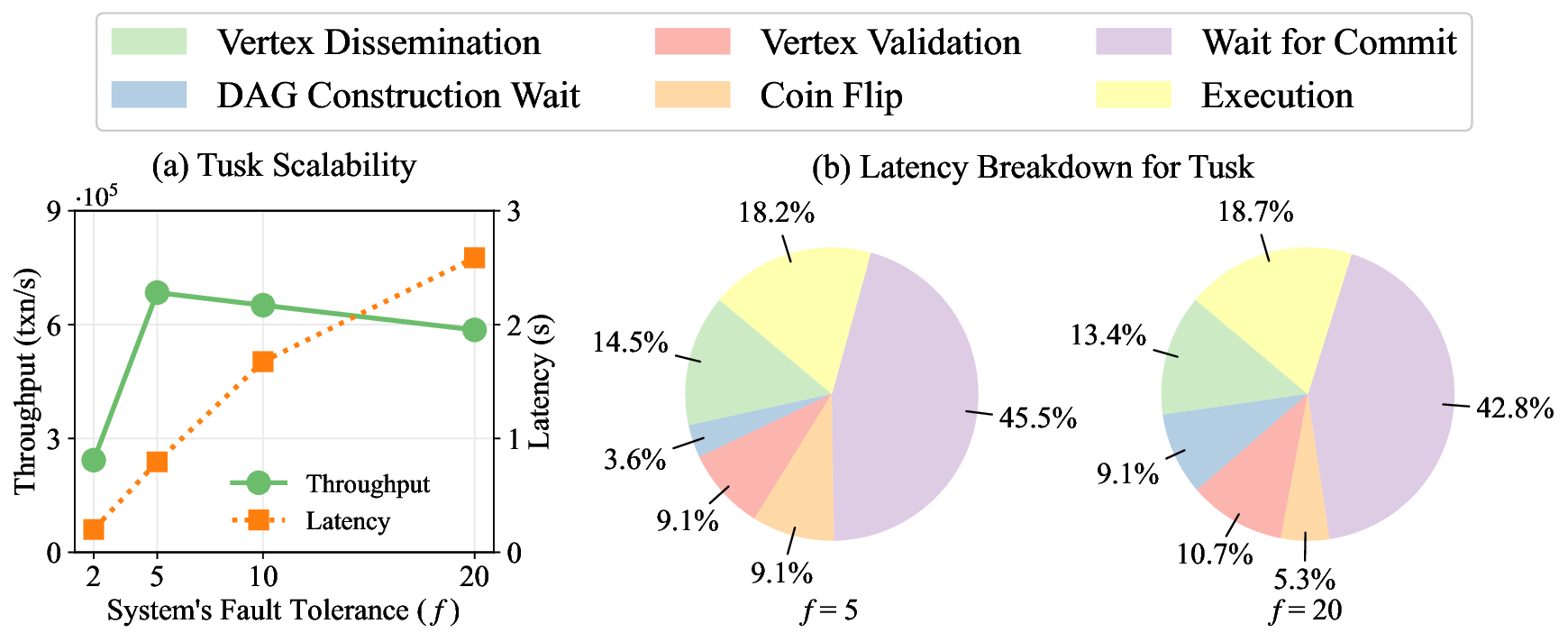}
    \vspace{-4mm}
    \caption{The Scalability and Latency Breakdown for Tusk. } \label{fig:motivation}
    \vspace{-6mm}
\end{figure}

These insights generalize across DAG protocols.
Building on them, \sysname employs TEE-assisted services to efficiently address the first four bottlenecks (\secref{sec:teedesign}):
\begin{packeditemize}
    \item We enhance fault tolerance to $n{=}2f{+}1$ via \textit{T-RBC} to improve scalability, while systematically analyzing and resolving the \textit{liveness challenges} arising from the quorum reduction.
    \item \textit{T-RBC} further optimizes dissemination complexity and latency.
    \item \textit{T-RoundCert} cuts crypto referencing complexity and overhead.
    \item \textit{T-Coin} accelerates secure randomness generation.
\end{packeditemize}

By utilizing the above four ideas in \sysname, Table~\ref{tab:dag-comparison} further summarizes how \sysname performs compared to representative DAG protocols in different aspects.
\section{System Model and Goals} \label{sec:systemmodel}

\subsection{System Model}
\bheading{Replicas with TEEs.}
We consider a network $\Pi$ consisting of $n{=}2f{+}1$ replicas, denoted by $\Pi=\{p_i\}_{i=1}^{n}$. 
Each replica $p_i$ has a unique identity known to all others within the network.
Each replica runs on an SGX-enabled machine that provides a Trusted Execution Environment (TEE), whose integrity and confidentiality are guaranteed by the underlying hardware. Up to $f$ replicas may be Byzantine and behave arbitrarily. The remaining replicas are honest and strictly follow the protocol.
In addition to replicas, there is an unbounded set of clients that create transactions.  

\bheading{Threat Model.} We adopt the threat model used in prior TEE-assisted BFT protocols~\cite{hybster, minbft, Damysus, dqbft, TBFT, Achilles}. We assume an adversary that controls up to $f$ Byzantine replicas and any number of clients.  
The adversary has root-level access to corrupted replicas except for TEEs, enabling control over the network. In other words, all components outside the TEEs are considered vulnerable to compromise. 
In contrast, the TEE itself is assumed to remain secure and untampered with, maintaining its core trusted functionalities.
We also assume that TEEs support secure remote attestation, providing unforgeable cryptographic evidence that a specific program is running inside enclaves~\cite{intelRA,amdRA}.
\new{Following prior TEE-assisted BFT work~\cite{hybster, minbft, Damysus, dqbft, TBFT, ENGRAFT}, we do not consider most TEE-specific attacks (\eg, transient execution~\cite{foreshadow, SgxPectre}, micro-architectural side channels~\cite{cacheout, crosstalk, sgaxe, aepic-leak}, fault injection~\cite{plundervolt}, and rollback~\cite{ROTE, Nimble, narrator, narrator-pro}), as these can be mitigated by microcode patches~\cite{intel_software_security_guidance}, hardened cryptographic libraries~\cite{CacheBleed}, and state-continuity mechanisms~\cite{ROTE, Nimble, ENGRAFT, narrator, narrator-pro}, respectively. Moreover, \sysname{}'s minimal TCB (only tiny components in the enclave) further reduces exposure to residual side channels compared to designs that place full protocol logic in enclaves.}

Each replica possesses a public/private key pair established through a Public-Key Infrastructure (PKI) and uses it to sign consensus messages (\eg, votes). The private keys are securely sealed and managed within TEEs, while the corresponding public keys are known to all replicas.  The adversary is computationally bounded and cannot compromise standard cryptographic primitives such as digital signatures or collision-resistant hash functions. All outputs generated by the enclave are cryptographically signed by the TEE, ensuring that they cannot be forged, modified, or tampered with by any replica, even if the host system is Byzantine.

\bheading{Asynchronous Network.} We assume an asynchronous network model~\cite{asyncnetwork}, where messages between replicas may experience arbitrary but finite delays, i.e., messages are eventually delivered and never lost, but there are no guarantees on delivery time or order.

\begin{table*}[t]
\vspace{-4mm}
\centering
\caption{Comparison of DAG-based BFT Protocols (Normalized to Message Delays $\delta$).}
\label{tab:dag-comparison}
\vspace{-3mm}
\setlength{\tabcolsep}{8pt}
\renewcommand{\arraystretch}{1.25}
\scalebox{0.85}{
\begin{tabular}{@{}lcccccc@{}}
\toprule[1pt]
 & Fault Tolerance & Network Assumption & LV Commit\textsuperscript{(1)} & Async LV Commit & NLV Commit\textsuperscript{(2)} & Bits per round\textsuperscript{(3)} \\
\midrule
DAG-Rider~\cite{DAGRider} & $3f{+}1$ & Async & $\!18\delta$ (6 RBCs) & $\!18\delta$ (6 RBCs) & ${+}\,7.5\delta$ & $\mathcal{O}(\kappa n^4)$ \\
Tusk~\cite{narwhal} & $3f{+}1$ & Async & $\!13.5\delta$ ($4.5$ RBCs) & $\!21\delta$ (7 RBCs) & ${+}\,4.5\delta$ & $\mathcal{O}(\kappa n^3)$ \\
Bullshark~\cite{Bullshark} & $3f{+}1$ & Async \& Partial Sync & $6\delta$ (2 RBCs) & $\!18\delta$ (6 RBCs) & ${+}\,4.5\text{--}7.5\delta$ & $\mathcal{O}(\kappa n^3)$ \\
\new{GradedDAG~\cite{gradeddag}} & \new{$3f{+}1$} & \new{Async} & \new{$5\delta$ (1 RBC \& 1 CBC)} & \new{$7.5\delta$ (1.5 RBC \& 1.5 CBC)} & \new{${+}\,5\delta$} & \new{$\mathcal{O}(\kappa n^3)$} \\
\new{Mahi-Mahi~\cite{jovanovic2024mahi}} & \new{$3f{+}1$} & \new{Async} & \new{$5\delta$ (5 broadcasts)} & \new{$5\delta/p^*$\textsuperscript{$\ddagger$}} & \new{${+}\,0\delta$} & \new{$\mathcal{O}(\kappa n^3)$} \\
Shoal++~\cite{shoal++} & $3f{+}1$ & Partial Sync & $\mathbf{4\delta}$ (1 RBC \& 1 broadcast) & $\infty$ & ${+}\,2\text{--}3\delta$ & $\mathcal{O}(\kappa n^3)$ \\
Sailfish~\cite{sailfish} & $3f{+}1$ & Partial Sync & $\,4\delta$ (1 RBC \& 1 broadcast) & $\infty$ & ${+}\,3\delta$ & $\mathcal{O}(\kappa n^3)$ \\
Mysticeti~\cite{babel2024mysticeti} & $3f{+}1$ & Partial Sync & $\mathbf{3\delta}$ (3 broadcasts) & $\infty$ & ${+}\,3\delta$ & $\mathcal{O}(\kappa n^3)$ \\
\rowcolor[gray]{0.9}
\textbf{\sysname} & $2f{+}1$ & Async & $4\delta$ (4 T-RBCs) & $8\delta$ (8 T-RBCs) & ${+}\,2.5\delta$ & $\mathcal{O}(\kappa n^2 + n^3)$ \\
\bottomrule[1pt]
\end{tabular}}
\vspace{1mm}

\raggedright
\footnotesize
\textsuperscript{(1)} LV = leader vertex. In leader–anchored commit rules (e.g., Bullshark~\cite{Bullshark}, Shoal++~\cite{shoal++}, Sailfish~\cite{sailfish}), blocks are decided with respect to the \emph{wave leader} of that wave~(\secref{sec:motivation}). \newline
\textsuperscript{(2)} NLV = non-leader vertex. In those leader-based commit rules, non-leader vertices are indirectly committed and typically need a few extra DAG rounds beyond leaders. \newline
\textsuperscript{(3)} We count \emph{metadata only} (parents/certificates/signatures), not payload. 
In recent DAG designs, each node emits one vertex to $\Theta(n)$ replicas per round with $n$ certified parents ($\Theta(\kappa)$ for \sysname, while $\Theta(\kappa n)$ for others). 
For DAG-Rider, the dissemination (RBC) complexity is $\Theta(n^2)$, while $\Theta(n)$ for all others.
In total $\Theta(n)$ nodes yield $\mathcal{O}(\kappa n^3)$ bits per round (constants depend on the signature scheme). \newline
\new{\textsuperscript{$\ddagger$} Under an async adversary, Mahi-Mahi commits probabilistically with $p^* {=} \ell/(3f{+}1)$ per wave ($\ell$~= leader slots/round). The expected async commit latency grows with $n$.}
\end{table*}

\subsection{DAG-Based Consensus Goals} \label{subsec:goals}

DAG-based BFT consensus protocols should allow replicas to agree on a sequence of client transactions in the presence of Byzantine behaviors (see the above threat model). 
Each replica continuously proposes blocks containing batches of transactions, with each block $b$ assigned a consecutive round number \(r\). A block is also a vertex in a DAG, and we use them interchangeably in this paper. 
The round number $r$ is used to distinguish between blocks proposed by the same replica.
Each block references at least \( f {+} 1 \) vertices from the previous round number, thus creating a DAG. 

Replicas then use common coins to select the leader vertex in a wave and further commit all blocks referenced by the leader vertex according to certain deterministic algorithms. Once a block is committed by a correct replica, it is assigned a unique sequence number $sn$, and replicas send the corresponding replies to clients.
The consensus protocol should satisfy the following properties: 
\begin{packeditemize} 
    \item \textbf{Safety:}  If two correct replicas commit two blocks $b$ and $b^{\prime}$ with the same sequence number $sn$, then $b=b^{\prime}$.
    \item \textbf{Liveness:} A transaction $tx$ created by a correct client will eventually be committed to a block $b$ by correct replicas. 
\end{packeditemize} 

\definecolor{lightred}{HTML}{F7C7C5}
\definecolor{lightblue}{HTML}{D5E4FB}
\definecolor{lightorange}{HTML}{FFD580}
\definecolor{darkyellow}{HTML}{dae465}
\definecolor{myblue}{HTML}{6182B7}
\definecolor{myred}{HTML}{AF4A47}

\begin{figure}[t]
\vspace{-3mm}
    \centering
    \begin{tikzpicture}[scale=0.77,
      node distance=0.25cm and 0.8cm,
      roundnode/.style={circle, draw, fill=white, minimum size=5mm},
      smallroundnode/.style={circle, draw, fill=white, minimum size=1mm},
      rednode/.style={circle, draw=myred, fill=lightred, minimum size=5mm},
      bluenode/.style={circle, draw=myblue, fill=lightblue, minimum size=5mm},
      smallbluenode/.style={circle, draw=myblue, fill=lightblue, minimum size=1mm},
      smallrednode/.style={circle,  draw=myred, fill=lightred, minimum size=1mm},
      orangenode/.style={circle, draw=myblue, fill=lightorange, minimum size=5mm},
      smallorangenode/.style={circle, draw=myblue, fill=lightorange, minimum size=1mm},
      edge/.style={thin, Stealth-},
      rededge/.style={thick, Stealth-, draw=myred},
      yellowedge/.style={thick, Stealth-, draw=darkyellow},
      blueedge/.style={thick, Stealth-, draw=myblue}
    ]

    \tikzset{>={Stealth[length=5pt,width=7pt]}}
    

    \draw[black, dashed, thick] (0.2,6) rectangle (8.3,9.1);
    \draw[orange, dashed, thick] (0,-0.2) rectangle (8.5,9.8);

    \draw[dashed, thick] (0.2,2.65) rectangle (8.3,5.8);
    \draw[dashed, thick] (0.2,-0.05) rectangle (8.3,2.45);

    \node[font=\footnotesize=] at (7.18, 8.88) {\textbf{Dissemination}};
    \node[font=\footnotesize] at (7.18, 8.5) {\textbf{Layer}};
    \node[font=\footnotesize] at (7.18, 2.18) {\textbf{Ordering}};
    \node[font=\footnotesize] at (7.18, 1.8) {\textbf{Layer}};

    \node at (1.6, 9.45) {\textbf{\sysname Overview}};
    \node[font=\tiny] at (1, 5) {\textbf{Replica$_1$}};
    \node[font=\tiny] at (1, 4) {\textbf{Replica$_2$}};
    \node[font=\tiny] at (1, 3) {\textbf{Replica$_3$}};
    \node[font=\footnotesize] at (7.25, 3.15) {\textbf{Local DAG}};




    \node at (1, 7.9) {\includegraphics[width = 0.5cm]{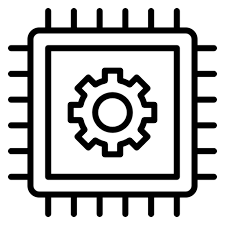}};
    \node at (1.2, 7.6) {\includegraphics[width = 0.25cm]{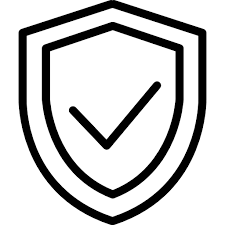}};
    \draw[dashed, blue, line width=1pt] (1,7.9) circle[radius=0.5cm];
    \node[font=\tiny] at (1, 7.2) {\textbf{RAC}};

    \node at (4, 7.9) {\includegraphics[width = 0.6cm]{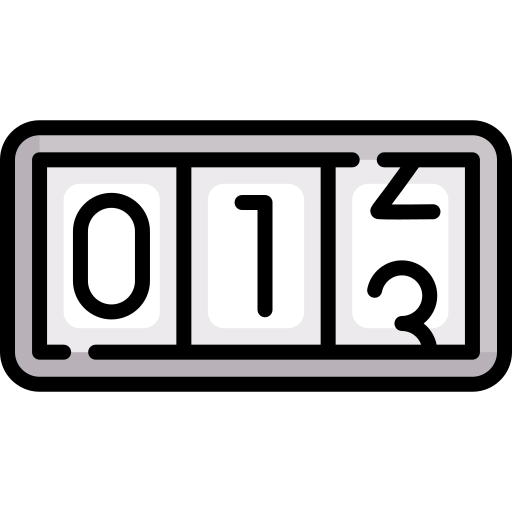}};
    \node at (4.2, 7.6) {\includegraphics[width = 0.25cm]{icons/shield.png}};
    \draw[dashed, orange, line width=1pt] (4,7.9) circle[radius=0.5cm];
    \node[font=\tiny] at (4, 7.2) {\textbf{MC}};

    \draw[->, gray, line width=2pt] (2, 6) -- (1, 7);

    \draw[->, gray, line width=2pt] (1.7, 7.8) -- (3.3, 7.8);
    \node[font=\tiny] at (2.5, 8.1) {round-cert};
    \node at (2.5, 8.5) {\includegraphics[width = 0.5cm]{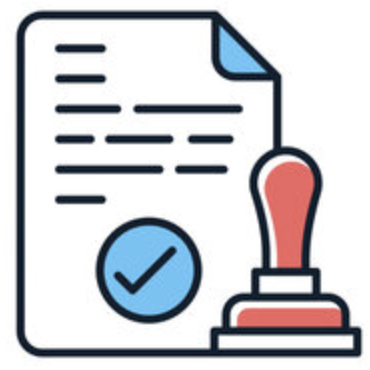}};

    \node at (5.1, 8.5) {\includegraphics[width = 0.55cm]{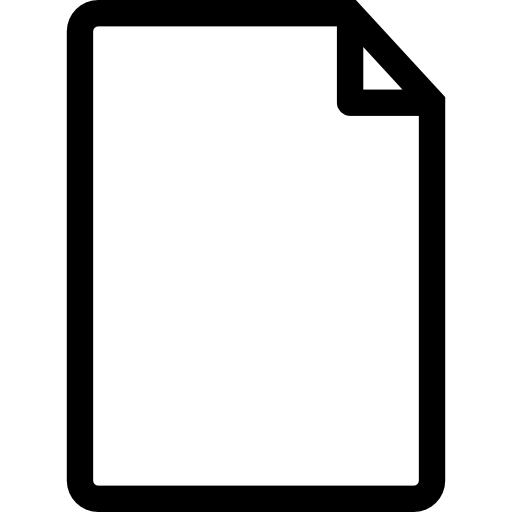}};
    \node at (5.3, 8.1) {\includegraphics[width = 0.2cm]{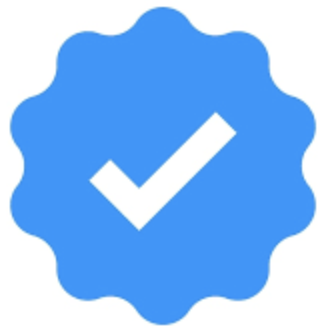}};
    \node[font=\tiny] at (5.1, 8.5) {$r\!:\!x$};
    \draw[->, gray, line width=2pt] (4.7, 7.8) -- (5.6, 7.8);
    \node[smallroundnode] (srn) at (6,7.8) {};

    \node[font=\tiny] at (6,7.5) {round-$x$};
    \node[font=\tiny] at (6,7.3) {vertex};

    \node at (7.3, 7.8) {\includegraphics[width = 0.7cm]{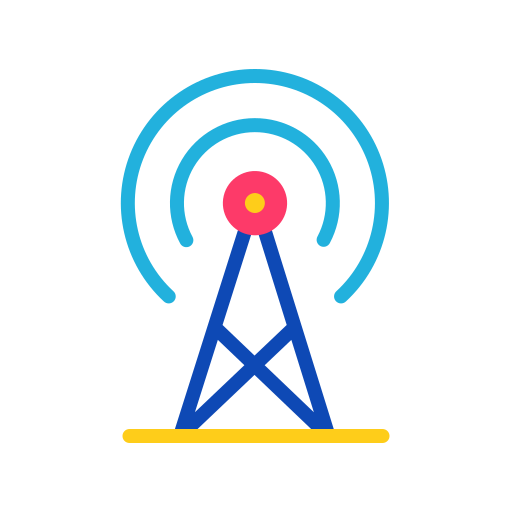}};
    \node at (7.9, 7.8) {\includegraphics[width = 0.5cm]{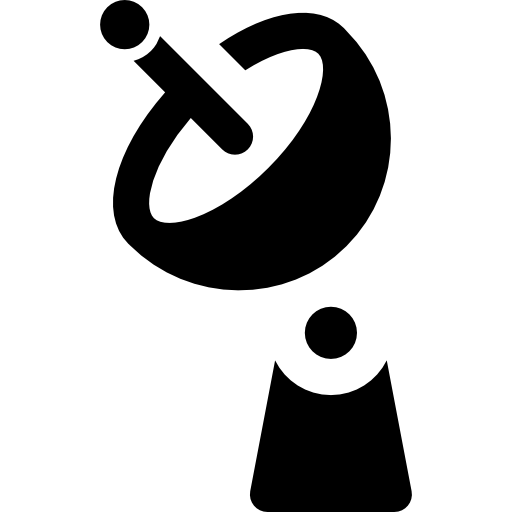}};
    \node[font=\tiny] at (7.5,7.1) {Dissemination};
    \draw[dashed, thick] (6.9,8.2) rectangle (8.2,7.3);
    \draw[->, gray, line width=2pt] (6.3, 7.8) -- (6.8, 7.8);
    \draw[->, gray, line width=2pt] (5.7, 7) -- (4.7,6);

    \node at (1.3, 0.85) {\includegraphics[width = 0.4cm]{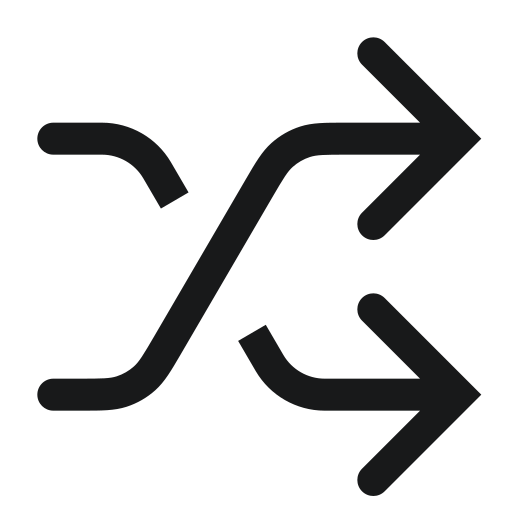}};
    \node at (1.5, 0.55) {\includegraphics[width = 0.25cm]{icons/shield.png}};
    \draw[dashed, red, line width=1pt] (1.3,0.85) circle[radius=0.5cm];

    \node[font=\tiny] at (1.3, 0.15) {\textbf{RNG}};
    \draw[->, gray, line width=2pt] (2.3, 2.4) -- (1.3, 1.4);
    \node[smallbluenode] (sr1) at (1.6, 2.1) {};
    \node[smallbluenode] (sr2) at (1.1, 2.1) {};

    \draw[->, gray, line width=2pt] (1.6, 1.25) -- (2.7, 2.35);
    \node[smallorangenode] (sb2) at (2.5, 1.55) {};
    \node[font=\tiny] at (2.5,1.25) {Leader};

    \node[font=\tiny] at (5.4, 0.35) {Execution};
    \node at (5.4, 0.95) {\includegraphics[width = 0.8cm]{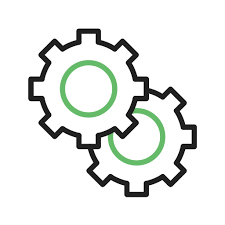}};
    
    \node[smallorangenode] (x1) at (4.3, 1.55) {};
    \node[smallrednode] (y1) at (3.6, 2.15) {};
    \node[smallroundnode] (y3) at (3.6, 1.55) {};
    \draw[edge] (y3) -- (x1);
    \draw[edge] (y1) -- (x1);
    \draw[->, gray, line width=2pt] (4, 2.4) -- (5, 1.4);
    \node[font=\tiny] at (4,1.1) {committed};
    \node[font=\tiny] at (4,0.9) {vertices};


    \node[font=\tiny] at (1.9,5.5) {$r=x-1$};
    \node[font=\tiny] at (3.3,5.5) {$r=x$};
    \node[font=\tiny] at (4.7,5.5) {$r=x+1$};
    \node[font=\tiny] at (6.1,5.5) {$r=x+2$};
    \node[font=\tiny] at (7.5,5.5) {$r=x+3$};
    \node[rednode] (11) at (1.9,5) {};
    \node[rednode] (21) at (1.9,4) {};
    \node[roundnode] (31) at (1.9,3) {};
    \node[roundnode] (12) at (3.3,5) {};
    \node[roundnode] (22) at (3.3,4) {};
    \node[orangenode] (32) at (3.3,3) {};
    \node[roundnode] (13) at (4.7,5) {};
    \node[roundnode] (23) at (4.7,4) {};
    \node[roundnode] (14) at (6.1,5) {};
    \node[roundnode] (24) at (6.1,4) {};
    \node[bluenode] (15) at (7.5,5) {};
    \node[bluenode] (25) at (7.5,4) {};
    
    \draw[edge] (11) -- (12);
    \draw[edge] (21) -- (12);
    \draw[edge] (31) -- (22);
    \draw[edge] (21) -- (22);
    \draw[edge] (21) -- (32);
    \draw[edge] (31) -- (32);

    \draw[edge] (12) -- (13);
    \draw[edge, color=myblue] (32) -- (13);
    \draw[edge, color=myblue] (32) -- (23);
    \draw[edge] (22) -- (23);
    \draw[edge, color=myblue] (13) -- (14);
    \draw[edge] (23) -- (14);
    \draw[edge] (13) -- (24);
    \draw[edge, color=myblue] (23) -- (24);
    \draw[edge, color=myblue] (14) -- (15);
    \draw[edge] (24) -- (15);
    \draw[edge] (14) -- (25);
    \draw[edge, color=myblue] (24) -- (25);

    \end{tikzpicture}
    \vspace{-3mm}
    \caption{Architecture Overview of \sysname.}
    \vspace{-2mm}
    \label{fig:architecture}
\end{figure}

\section{\sysname Overview} \label{sec:system_design}


Fig.~\ref{fig:architecture} gives an architectural overview of \sysname, an asynchronous DAG-based protocol that tolerates $f$ Byzantine faults among $n{=}2f{+}1$ replicas, while Fig.~\ref{fig:commit} illustrates a sample instantiation of the ordering layer.
Specifically, \sysname operates in two distinct layers: 

\begin{packeditemize}
    \item \textbf{Dissemination layer:} In this layer (\secref{subsec:dag_construction}), replicas construct their local views of DAG by a prescribed construction procedure. With the support of TEEs, replicas utilize the \textit{T-RBC} service (\secref{subsec:mic}) to reliably disseminate vertices. To ensure that replicas have consistent local views of the DAG, the \textit{T-RoundCert} service (\secref{subsec:rac}) is employed to validate references between vertices.  
    
    \item \textbf{Ordering layer:} In this layer (\secref{subsec:commit}), each replica independently interprets its local view of the DAG to deterministically establish an order of vertices. To ensure that all honest replicas derive a consistent order, they follow a four-round commit rule (\secref{sec:commitrule}). With the support of TEEs, we introduce the \textit{T-Coin} mechanism (\secref{subsec:rang}) to generate shared randomness that determines the vertex order consistently across replicas, which could eliminate the additional coordination required in prior DAG-based protocols.

\end{packeditemize}

\usetikzlibrary{decorations.pathreplacing, positioning}

\begin{figure}[t!]
\vspace{-4mm}
\centering
\hspace*{-0.35cm}
\scriptsize
\begin{tikzpicture}[
    scale=0.7,
    node distance=0.25cm and 0.44cm,
    roundnode/.style={circle, draw, fill=white, minimum size=6.5mm},
    rednode/.style={circle, draw=myred, fill=lightred, minimum size=6.5mm},
    bluenode/.style={circle, draw=myblue, fill=lightblue, minimum size=6.5mm},
    edge/.style={thin, Stealth-},
    rededge/.style={thick, Stealth-, draw=myred},
    yellowedge/.style={thick, Stealth-, draw=darkyellow},
    blueedge/.style={thick, Stealth-, draw=myblue}
]

\node[roundnode] (a1) at (0,-1) {};
\node[roundnode] (a2) [below=of a1] {};
\node[rednode] (a3) [below=of a2] {};
\node[above=0.03cm of a1] {$r=x$};
\node[below=0.25cm and 0.1cm of a3, text=myred] (desc) {Coin revealed};
\draw[->, bend left, draw=myred] (desc) to (a3);

\node[roundnode] (b1) [right=of a1] {};
\node[roundnode] (b2) [below=of b1] {};
\node[roundnode] (b3) [below=of b2] {};
\node[above=0.03cm of b1] {$r=x+1$};

\node[roundnode] (c1) [right=of b1] {};
\node[roundnode] (c2) [below=of c1] {};
\node[roundnode] (c3) [below=of c2] {};
\node[above=0.03cm of c1] {$r=x+2$};

\node[roundnode] (d1) [right=of c1] {};
\node[roundnode] (d2) [below=of d1] {};
\node[roundnode] (d3) [below=of d2] {};
\node[above=0.03cm of d1] {$r=x+3$};

\node[below=0.08cm of d3, shift={(0.05cm,-0.05cm)}] {\textit{Coin to elect the leader in the last round of every wave}};
\node[below=0.08cm of d3, shift={(0.05cm,-0.3cm)}] {\textit{(\ie, $r=x+7$ reveal the leader of wave $w+1$)}};

\node[roundnode] (e1) [right=of d1] {};
\node[bluenode] (e2) [below=of e1] {};
\node[roundnode] (e3) [below=of e2] {};
\node[above=0.03cm of e1] {$r=x+4$};
\node[below right=0.8cm and 0.1cm of e2, shift={(1.3cm,-0.3cm)}, text=myblue] (desc) {Coin revealed};
\draw[->, bend left, draw=myblue] (desc) to (e2);

\node[roundnode] (f1) [right=of e1] {};
\node[roundnode] (f3) [below=1.15cm of f1] {};
\node[above=0.03cm of f1] {$r=x+5$};

\node[roundnode] (g1) [right=of f1] {};
\node[roundnode] (g3) [below=1.15cm of g1] {};
\node[above=0.03cm of g1] {$r=x+6$};

\node[roundnode] (h1) [right=of g1] {};
\node[roundnode] (h3) [below=1.15cm of h1] {};
\node[above=0.03cm of h1] {$r=x+7$};

\draw[thin, decorate, decoration={brace, amplitude=10pt}] 
    ([shift={(-0.1,0.5)}]a1.north west) -- ([shift={(0.1,0.5)}]d1.north east)
    node[midway, above=10pt] {\textit{Wave w}};

\draw[thin, decorate, decoration={brace, amplitude=10pt}] 
    ([shift={(-0.1,0.5)}]e1.north west) -- ([shift={(0.1,0.5)}]h1.north east)
    node[midway, above=10pt] {\textit{Wave w+1}};

\draw[edge] (a1) -- (b1);
\draw[edge] (a1) -- (b2);
\draw[edge] (a1) -- (b3);
\draw[edge] (a2) -- (b1);
\draw[edge] (a2) -- (b2);
\draw[edge] (a2) -- (b3);
\draw[rededge] (a3) -- (b3);

\draw[edge] (b1) -- (c1);
\draw[edge] (b1) -- (c2);
\draw[edge] (b1) -- (c3);
\draw[edge] (b2) -- (c1);
\draw[edge] (b2) -- (c2);
\draw[edge] (b2) -- (c3);
\draw[rededge] (b3) -- (c3);

\draw[edge] (c1) -- (d1);
\draw[edge] (c1) -- (d2);
\draw[edge] (c1) -- (d3);
\draw[edge] (c2) -- (d1);
\draw[edge] (c2) -- (d2);
\draw[edge] (c2) -- (d3);
\draw[rededge] (c3) -- (d3);
\draw[yellowedge] (d3) -- (e2);

\draw[edge] (d1) -- (e1);
\draw[edge] (d1) -- (e2);
\draw[edge] (d1) -- (e3);
\draw[edge] (d2) -- (e1);
\draw[edge] (d2) -- (e2);
\draw[edge] (d2) -- (e3);
\draw[edge] (d3) -- (e3);

\draw[edge] (e1) -- (f1);
\draw[edge] (e1) -- (f3);
\draw[blueedge] (e2) -- (f1);
\draw[blueedge] (e2) -- (f3);
\draw[edge] (e3) -- (f3);

\draw[blueedge] (f1) -- (g1);
\draw[edge] (f1) -- (g3);
\draw[edge] (f3) -- (g1);
\draw[blueedge] (f3) -- (g3);
\draw[blueedge] (g1) -- (h1);
\draw[edge] (g1) -- (h3);
\draw[edge] (g3) -- (h1);
\draw[blueedge] (g3) -- (h3);

\end{tikzpicture}

\vspace{-4mm}
\caption{Example of Commit Rule in \sysname. If a leader has less than $f{+}1$ references from its fourth round (red in the figure), it is ignored; otherwise (blue in the figure), the algorithm searches the causal DAG to commit all preceding leaders (dark yellow and red in the figure) and then orders the remaining DAG by the predefined rule afterward.} 
\vspace{-4mm}
\label{fig:commit}
\end{figure}


\bheading{Consensus Flow.} 
We now present a high-level processing flow.

\vspace{1mm} \noindent \blackding{1} \textbf{Client Transaction Submission.} Clients submit transactions to replicas, which verify and mark the transactions as pending for inclusion into vertices.
  
\vspace{1mm} \noindent \blackding{2} \textbf{Vertex Construction.} \sysname progresses in a round-by-round manner, where each replica proposes one vertex per round. Upon entering a new round and collecting at least $f{+}1$ vertices from the previous round, a replica constructs a new vertex that includes a batch of pending transactions, a validity proof, and relevant metadata. The validity proof is generated by the \textit{T-RoundCert} service based on $f{+}1$ references to vertices from the preceding round.

\vspace{1mm} \noindent \blackding{3} \textbf{Vertex Dissemination.} Once the vertex is constructed, it is disseminated to all replicas through the efficient \textit{T-RBC} service.

\vspace{1mm} \noindent \blackding{4} \textbf{DAG Construction.} Upon receiving a vertex through the \textit{T-RBC} service, a replica first verifies the validity of the vertex. If any referenced vertices in the validity proof are missing, the replica waits until they are received and validated. Once all dependencies are resolved, the vertex is appended to the local DAG. 

\vspace{1mm} \noindent \blackding{5} \textbf{Wave Leader Election.}
Each replica interprets its local DAG to determine the order of vertices in a wave-by-wave manner. In the fourth round of each wave, replicas leverage the shared randomness provided by the \textit{T-Coin} service to collaboratively elect a leader vertex (i.e., a vertex from the first round) for that wave.

\vspace{1mm} \noindent \blackding{6} \textbf{Vertex Ordering and Commitment.} Once a leader vertex is elected and satisfies the commit rule—specifically, when at least $f{+}1$ vertices from the fourth round of the same wave have a path to it—the leader vertex, together with all vertices in its causal history, is committed in a predetermined order.

Specifically, we observed that reducing the system size to $n{=}2f{+}1$ using the \textit{T-RBC} service introduces a potential \textit{liveness} issue such that the leader vertex can never satisfy a traditional two-round commit rule, which drives us to develop the four-round wave and commit rule. This four-round design and the corresponding commit rule preserve both safety and liveness while maintaining low commit latency, especially under adversarial asynchronous network conditions (see Sec.~\ref{sec:correctness_analysis} for details).

\vspace{1mm} \noindent \blackding{7} \textbf{Transaction execution.} Replicas execute the transactions contained in committed vertices and return the results to the client.

Steps \blackding{2}--\blackding{4} and \blackding{5}--\blackding{6} constitute the Dissemination Layer (\secref{subsec:dag_construction}) and Ordering Layer (\secref{subsec:commit}), respectively.  

\section{Trusted Components within TEEs} \label{sec:teedesign}
\subsection{TEE-Assisted Reliable Broadcast}~\label{subsec:mic}  

\vspace{-4mm}
\noindent \textbf{Trusted Component - MC.} The \textit{Monotonic Counter (MC)} is a trusted component designed to prevent message equivocation in RBC design.
In each round, the \textit{MC} gives its vertex a unique, verifiable identifier implemented by a monotonically increasing counter in TEE, ensuring that each replica can generate at most one valid vertex per round.
Its interface is as follows:

\begin{packeditemize}
    \item $\langle counter\rangle\gets\Name{GetCounter}\textit{(round-cert, v)}$:
    Given a vertex \textit{v} and a round certificate \textit{round-cert} which certifies that the replica has received at least $f{+}1$ valid vertices from the previous round (see~\secref{subsec:rac}), this function returns a fresh counter value, attaches it to \textit{v} along with a cryptographic proof, and atomically increments the counter by one. The counter value is set as 0 at initialization.
 \end{packeditemize}

\bheading{Reliable Broadcast with \textit{MC}.} Traditional RBC protocols (\eg, Bracha’s RBC~\cite{brachaRBC}) require three rounds of all-to-all communication and $n{=}3f{+}1$ replicas to tolerate $f$ Byzantine faults, resulting in quadratic message complexity and significant communication overhead (\secref{subsec:dissect}).
In contrast, \sysname integrates \textit{MC} to construct a \textit{TEE-assisted RBC (T-RBC)} where non-equivocation can be guaranteed with the cryptographically signed \textit{MC} value in each vertex. Thus, without the need to collect support from most honest replicas, \textit{T-RBC} has linear communication complexity and a lower system size of $n{=}2f{+}1$.

Since \sysname’s RBC operates on top of a DAG, we further leverage the causality of the DAG to resolve potential issues in which a vertex is received only by part of the honest replicas. 
In the DAG, vertices reference vertices from earlier rounds. When vertices are missing from the causal history of a newly received vertex $v$, a replica proactively requests them by sending a catch-up request message to the sender of $v$, with linear message complexity. The complete workflow of T-RBC is detailed in \apxref{app:subsec:trbc-flow}.

We formally prove T-RBC’s correctness in \apxref{app:proof:tee_rbc}. Overall, T-RBC achieves linear communication complexity, one-step communication latency in the good case, and operates with a reduced system size of $n{=}2f{+}1$ replicas.

\subsection{TEE-Assisted Round Certifier}~\label{subsec:rac}

\vspace{-4mm}
\noindent \textbf{Trusted Component - RAC.} The \textit{Round Advancement
Certifier (RAC)} accelerates round certification and vertex validation in the Dissemination Layer.
When a replica generates a vertex, it must reference a quorum of at least $f{+}1$ vertices from the previous round (\secref{subsec:dag_construction}).
To validate the references, RAC checks whether each is delivered through \textit{T-RBC} by verifying the \textit{MC} value within the vertex.
Upon successful validation of $f{+}1$ references to vertices from the previous round, RAC issues and cryptographically signs a \textit{round-cert} that is encoded as an $n$-bit bitmask, justifying its eligibility to propose a vertex in the current round.

Upon receiving a vertex, a replica can verify its validity by verifying the attached round certificate, ensuring that the vertex refers to $f{+}1$ vertices of the previous round.

This allows the replica to replace the $n-f$ hash references of size $\theta(\kappa n)$ with a compact certificate of size $\theta(\kappa+n)$ and reduce the number of cryptographic verifications from $\theta(n)$ to $\theta(1)$, substantially reducing CPU overhead.
RAC provides the following interface:
\begin{packeditemize}
\item $\langle$\textit{round-cert}$\rangle \gets \Name{ValidateVertices}(\textit{vertexList})$:
Given a list of referenced vertices, this function verifies whether the vertices have been delivered through \textit{T-RBC}, ensures that there are at least $f{+}1$ vertices, and validates their metadata correctness. If verified, it returns a round certificate $\langle$\textit{round-cert}$\rangle$.
The certificate is an $n$-bit mask indicating which vertices from the previous round are referenced (bit $i = 1$ iff the $i$-th vertex is included). 
\end{packeditemize}

\subsection{TEE-Assisted Common Coin}~\label{subsec:rang}

\vspace{-4mm}
\noindent \textbf{Trusted Component - RNG.} In the Ordering Layer, the random number generator (RNG) replaces the cryptographic common coin with a lightweight, TEE-assisted alternative for leader vertex selection. Previous DAG protocols implement the common coin via threshold cryptography, which is costly~\cite{thresholdlatency}.
\textit{RNG} instead uses a shared seed inside each TEE to produce identical, verifiable leader vertices across correct replicas. At startup, a DKG protocol~\cite{Shamirsecret} seeds \textit{RNG} consistently across all TEEs.

The leader vertices are generated wave by wave, where each wave consists of four rounds, and the leader vertex is in the first round. To request the random leader vertex value of a wave, a valid round certificate $\langle$\textit{round-cert}$\rangle$ for the fourth round, issued by the \textit{RAC}, is required. This design ensures that the random leader vertex becomes available only after a quorum of at least $f{+}1$ replicas have proposed vertices in the fourth round.

The four-round commit rule determines commitment by checking whether a sufficient number of fourth-round vertices have paths to the leader vertex (\secref{sec:commitrule}). 
If Byzantine nodes were able to know the leader vertex in advance, the adversary could undermine liveness by deliberately excluding the leader vertex from their causal histories and manipulating message delivery times to reduce the likelihood of its commitment. 
\textit{RNG} thus preserves the security and unpredictability of the traditional common coin while reducing computation to a single trusted call.
Its correctness is formally proved in \apxref{app:proof:tee_coin}. The \textit{RNG} has the following interfaces:
\begin{packeditemize}
\item $\langle \rho_r \rangle \gets \Name{Rand}(\textit{round-cert})$:
Given a valid \textit{round-cert} for round~$r+3$, returns the corresponding leader vertex $\rho_r$ of round $r$ if the \textit{round-cert} is verified.
\end{packeditemize}

\section{Dissemination Layer} \label{subsec:dag_construction}
We first introduce the basic structure of the DAG, and then describe its construction, which forms the dissemination layer of \sysname.

\begin{figure}[t!]
    \vspace{-3mm}
    \begin{myprotocol}
    \setcounter{ALC@line}{0}

    \SPACE \textbf{\underline{Data Structures (for replica $p_i$)}}

    \SPACE\quad \textbf{Vertex $v$:}
    \SPACE\quad\quad $v.\textit{round}$ --- the round index of $v$ in the DAG
    \SPACE\quad\quad $v.\textit{source}$ --- the replica that broadcasts $v$
    \SPACE\quad\quad $v.\textit{txns}$ --- the batch of transactions
    \SPACE\quad\quad $v.\textit{strongEdges}$ --- a bitmask ($\textit{round-cert}$) indexed by vertices in round $v.\textit{round}-1$ representing \textit{strong edges}
    \SPACE\quad\quad $v.\textit{weakEdges}$ --- a list of hashes of vertices in rounds $< v.\textit{round}-1$ representing \textit{weak edges}
    \SPACE\quad\quad $v.\textit{counter}$ --- the signed \textit{MC} value

    \SPACE\quad \textbf{In-Memory Variables:}
    \SPACE\quad\quad $DAG_i$ --- the local view of DAG of $p_i$, a map from round number $r$ to the set of vertices known to replica $p_i$ in round $r$
    \SPACE\quad\quad $txnsToPropose$ --- a queue of client transactions that $p_i$ received and has not proposed

    \SPACE
    \FUNCTION{$\textsc{Path}$}{$v, u$}
        \RETURN true if there exists a path of $v$ to $u$ in $DAG_i$. 
    \ENDFUNCTION

    \SPACE
    \FUNCTION{$\textsc{StrongPath}$}{$v, u$}
        \RETURN true if there exists a path of $v$ to $u$ in $DAG_i$ such that all edges on the path are strong edges.
    \ENDFUNCTION


    \end{myprotocol}

    \vspace{-3mm}
    \caption{Data Structures and Basic Utilities for Replica $p_i$.}
    \vspace{-5mm}
    \label{algorithm:utilities}
\end{figure}

\subsection{DAG Vertex Structure}
\figref{algorithm:utilities} shows the pseudocode for the DAG data structure and its basic utilities. Each vertex in the DAG contains: (i) associated metadata (\eg, round number, source), (ii) a batch of transactions, (iii) references to vertices in previous rounds\footnote{We ignore the \textit{MC} value here as it is calculated during \textit{T-RBC}.}. The references are divided into two types:

\begin{packeditemize}
\item \textbf{Strong Edges:} a bitmask (\textit{round-cert}) indexed by vertices showing references to at least $f{+}1$ vertices from the immediately preceding round $v.round-1$.

\item \textbf{Weak Edges:} a list of hashes showing references to vertices from rounds lower than $v.round-1$. 
\end{packeditemize}


\begin{figure}[!t] \label{algorithm:dagconstuction}
    \vspace{-3mm}
    \begin{myprotocol}
    \setcounter{ALC@line}{6}
    \TITLE{DAG Construction}{running at each replica$_i$}
    \SPACE \textbf{Local variables:}\quad $r \GETS 0$; $buffer \GETS \{\}$
    \SPACE Suppose each replica has a genesis vertex $v_0$

    \SPACE
    \EVENT[Receive $v$ from T-RBC]{$Deliver_i(v, round,p_k)$} \label{line:r_deliver}
    \IF[Verify the vertex]{\Name{VerifyVertex($v$)}} \label{line:vertex_validity}
    \IF{$\forall v' \in v.strongEdges \cup v.weakEdges: v' \in \bigcup_{k\geq1}DAG_i[k]$} \label{line:vertices_in_DAG} 
    \STATE $DAG_i[v.round] \GETS DAG_i[v.round] \cup \{v\}$
    \ELSE
        \STATE \Name{Defer} $v$ until any new vertex committed
    \ENDIF
    \ENDIF
    \ENDEVENT

    \SPACE
    \EVENT{$\lvert DAG_i[r] \rvert \geq f{+}1$} \label{line:advance_round}
    \SPACE \COMMENT{The quorum is adjusted to $f{+}1$ out of $2f{+}1$}
    \IF{\textcolor{black}{$r>1\wedge(r-1)\ mod\ 4 = 0$}} 
    \STATE \Name{WaveReady} ($\dfrac{r-1}{4}$) \COMMENT{Each wave contains 4 rounds}
    \ENDIF
    \STATE $\textit{round-cert} \GETS RAC.\Name{ValidateVertices}(DAG_i[r])$ \label{line:RAC}
    \STATE $r \GETS r + 1$;\quad $v \GETS$ \Name{CreateNewVertex}($r, \textit{round-cert}$) \label{line:generate_vertex}
    \STATE $\textsc{TBcast}_i(v,r)$ \label{line:rbc}
    \ENDEVENT
    
    \SPACE
    \FUNCTION{CreateNewVertex}{$r, \textit{round-cert}$}
    \STATE $v.txns$ \GETS ~$txnsToPropose$
    \STATE $v.strongEdge$ \GETS ~$DAG_i[r-1]$ \label{strongedge}
    \STATE \Name{SetWeakEdges}($v$,$r$)
    \RETURN $v$
    \ENDFUNCTION
    
\SPACE

    \FUNCTION{SetWeakEdges}{$v$,$r$} \label{weakedge}
    \STATE Add every prior vertex not reachable in DAG as weak edges
    \ENDFUNCTION
    \end{myprotocol}
    \vspace{-3mm}
    \caption{Pseudocode for DAG Construction.} 
    \vspace{-3mm}
    \label{algorithm:dag_construction}
\end{figure}

\subsection{DAG Construction}

\sysname proceeds in a round-by-round manner. Once a replica $p_i$ receives $f{+}1$ valid vertices from round $r$, it advances to round $r{+}1$, forming a new vertex $v$ referencing these vertices using strong edges and including a \textit{round-cert} of round $r$, and then broadcasts $v$, incrementally constructing the DAG.

\figref{algorithm:dag_construction} illustrates the DAG construction in \sysname, including vertex creation, verification, and integration into the DAG. Replicas begin with the genesis vertices and progress round by round. A replica advances to the next round once it has received at least $f{+}1$ vertices in the current round (line~\ref{line:advance_round}).

To ensure correctness, \sysname employs the \textit{T-RoundCert} mechanism. Before moving to the next round, a replica $p_i$ invokes \textit{RAC} to validate quorum criteria (line~\ref{line:RAC}), specifically that $f{+}1$ vertices from distinct replicas have been delivered. Upon successful validation, \textit{RAC} issues a round certificate (\textit{round-cert}), authorizing $p_i$ to advance. Then, $p_i$ generates a new vertex (line~\ref{line:generate_vertex}) and disseminates it to all replicas via \textit{T-RBC} (line~\ref{line:rbc}). The TEE-assisted services ensure that all vertices meet the vertex structure requirements and are delivered reliably, preserving the causal structure of the DAG.

Each replica maintains its own local view of the DAG. Upon receiving a vertex, the replica fetches all missing referenced vertices before adding the new vertex (line~\ref{line:vertices_in_DAG}). 
Although temporary divergences may occur due to network asynchrony, \textit{T-RBC} guarantees that all non-faulty replicas eventually receive the same set of vertices (Lemma~\apxlmref{lm:DAGview}{app:reliable_diss}). 
The differences in delivery order do not affect correctness because the DAG’s topological order is determined by vertex references, as shown in the example below. 

\begin{figure}[t]
\vspace{-3mm}
\begin{myprotocol}
\setcounter{ALC@line}{29}
\TITLE{\sysname}{running at each replica $p_i$}
\SPACE \textbf{Local variables:}
\SPACE $decidedWave$ \GETS $0$; $deliveredVertices, leadersStack$ \GETS ~\{\}

\SPACE

\EVENT{On receiving a block $b$ from clients}
\STATE $txnsToPropose$.\Name{Enqueue}($b$)
\ENDEVENT

\SPACE
\FUNCTION{WaveReady}{$w$} \label{line:wavereadystart}
\STATE $v$ \GETS ~\Name{GetWaveVertexLeader}($w$)
\IF{$v= \ \perp \vee \ \lvert \{v' \in DAG_i[round(w,4)]: \Name{StrongPath}(v',v)\}\rvert < f{+}1$} \label{f+1commit}
\RETURN
\ENDIF
\STATE $leadersStack$.\Name{Push}($v$) \label{line:pushleader1}
\FOR{wave $w'$ from $w$ - 1 down to $decidedWave$ + 1} \label{line:recursivecheck}
\STATE $v' \GETS \Name{GetWaveVertexLeader}(w')$
\IF{$v'\neq\ \perp \wedge\ \Name{StrongPath}(v,v')$}
\STATE  $v \GETS v'$;\quad $leadersStack$.\Name{Push}($v'$) \label{line:pushleader2}
\ENDIF
\ENDFOR
\STATE $decidedWave \GETS w$
\STATE \Name{OrderVertices}($leadersStack$)
\ENDFUNCTION \label{line:wavereadyend}

\SPACE


\FUNCTION{GetWaveVertexLeader}{$w$}
\STATE $j \GETS RNG.\Name{Rand}_i(v.$\textit{round-cert}$)$ \label{line:rand}
\IF{$\exists v \in DAG_i[round(w,1)] ~s.t.\ v.source = p_j$}
\RETURN $v$
\ENDIF
\RETURN $\perp$
\ENDFUNCTION

\SPACE


\FUNCTION{OrderVertices}{$leadersStack$} \label{line:ordervertices_begin}
\WHILE{$\lnot leadersStack$.\Name{IsEmpty}()}
\STATE $v \GETS leadersStack$.\Name{Pop}() \label{line:popleader}
\STATE $ToDeliver \GETS \{v' \in \bigcup_{r>0} DAG_i[r] ~\lvert$\\ ~~~~~~~~~~~~~~~~~~~~~~~ \Name{Path}$(v,v') \wedge v' \notin deliveredVertices\}$
\FOR{\textbf{every} $v' \in ToDeliver$ in predefined order} \label{line:predefined_order}

\STATE \textbf{output} $\Name{CommitBlock}_i(v')$
\STATE $deliveredVertices \GETS deliveredVertices ~\bigcup~ \{v'\}$
\ENDFOR
\ENDWHILE \label{line:ordervertices_end}
\ENDFUNCTION 


\end{myprotocol}
\vspace{-3mm}
\caption{Pseudocode of \sysname's Byzantine Consensus.} \label{algorithm:main}
\vspace{-3mm}
\end{figure}

\section{Ordering Layer} \label{subsec:commit}
\sysname requires no additional communication for transaction commitment, consistent with existing DAG-based protocols~\cite{DAGRider, narwhal, Bullshark}. Vertices are organized into \textit{waves}, each consisting of $k$ \textit{rounds} of vertex generation.
Within each wave, the rounds collectively form a \textit{common core} problem—determining how many vertices from the first round can be safely committed. To ensure liveness, at least one vertex must satisfy the commitment condition. Moreover, the number of such vertices determines the probability that a randomly chosen leader vertex can be successfully committed, which directly impacts the expected latency. The optimal number of rounds per wave, $k$, is derived analytically in \secref{sec:bound}.

At the final round of each wave, \sysname invokes a lightweight \textit{T-Coin} to randomly select a leader vertex among those generated in the first round. Each replica independently verifies whether the leader vertex belongs to the \textit{common core}. If the condition holds, all replicas deterministically order the vertices within the wave and commit the corresponding transactions. The complete commit rule, including how replicas apply the coin outcome to resolve ties and maintain consistent views, is detailed in \secref{sec:commitrule}.

\subsection{Dissecting Liveness with $f{+}1$ Quorums} \label{sec:bound}
We now derive the optimal number of rounds $k$ that minimizes the expected commit latency.
First, as said above, the rounds in a wave are abstracted as a \textit{$k$-iteration common core problem}, which generalizes the original common core formulation used in prior DAG consensus protocols (e.g., Tusk~\cite{narwhal}, where $k=2$).
Notably, the original common core was designed for standard Byzantine settings, where the threshold is $2f{+}1$ among $3f{+}1$ replicas.
In contrast, \sysname leverages TEE-assisted trust to reduce the threshold to $f{+}1$ among $2f{+}1$ replicas.
This change invalidates a direct reuse of existing results from the original common core formulation.

\subsubsection{$k$-Iteration Common Core Abstraction} \label{sec:core}
We first use $x_i$ to denote the input of replica $p_i$, \ie, the vertex it generates in the first round of a wave in \sysname.
Let $N$ denote the total number of participants in the common core instance, equal to the number of replicas in \sysname.
Throughout, we fix $N = 2f{+}1$ as in \secref{sec:systemmodel}.
The protocol proceeds in $k$ conceptual rounds, where each round builds upon the aggregation results of the previous one. 

\begin{packedenumerate}  
    \item \textbf{Initial Broadcast ($r = 1$).}  
    Each party broadcasts its input $x_i$ using the Reliable Broadcast protocol (RBC).

    \item \textbf{Recursive Aggregation (for $r = 2, 3, \ldots, k$).}  
    Each party $i$ maintains a set $A_i^r$ representing its aggregation result at round $r$.

    \begin{enumerate}
        \item For $r = 2$: define $A_i^2 = \{(j, x_j)\}.$ Once $|A_i^2| \ge f{+}1$, send $A_i^2$ to all parties.
        
        \item For $r > 2$: upon receiving $A_j^{r-1}$ from party $j$, accept it after receiving all broadcasts $x_k$ for every $(k, x_k) \in A_j^{r-1}$.  
        After accepting $f{+}1$ such sets, compute and broadcast $A_i^r = \bigcup_j A_j^{r-1}.$
    \end{enumerate}

    \item \textbf{Output.}  
    After completing $k$ aggregations, output $U_i = A_i^k.$
\end{packedenumerate}

The $k$-iteration common core protocol satisfies three properties:

\begin{packeditemize}
    \item \textbf{Common Core.}
    There exists a core set $S^*$ ($|S^*| > 0$) such that every non-faulty party includes $S^*$ in its output set.

    \item \textbf{Validity.}
    If a non-faulty party includes a pair $(j, x_j)$ in its output and party $j$ is non-faulty, then $x_j$ must be $j$’s original input.

    \item \textbf{Agreement.}
    All parties that include a pair for some sender $j$ must agree on its value.
    Formally, if two non-faulty parties include $(j, x)$ and $(j, x')$ in their outputs, then $x = x'$.
\end{packeditemize}

The agreement and validity properties can be satisfied by the RBC protocol (\apxref{app:build_block}). 
The number of rounds $k$ determines the size of $S^*$, and the maximum size of $S^*$ is bounded by $f{+}1$, since the $f$ adversaries can maliciously avoid generating vertices in certain rounds.

\bheading{Bounds of $|S^*|$.} To ensure system liveness, we must guarantee that $|S^*| \geq 1$.
Since only the leader vertex is selected from $S^*$, a wave can be committed only if this condition holds. If $|S^*| = 0$ under some adversarial scenario, the expected latency diverges to infinity, indicating that no wave can ever be committed.

\bheading{Expected Latency.} The size of $S^*$ determines the probability of successful commitment, further influencing the expected commit latency.
We now derive the expected commit latency of \sysname.

\begin{lemma}[Expected Latency]\label{lm:expected_latency}
    The expected latency of committing a wave is $ \frac{k \times N}{|S^*|}$, where $\frac{|S^*|}{N}$ denotes the expected probability of committing a single wave.
\end{lemma}

\begin{proof}
In the $k$-iteration common core abstraction, each round involves communication and aggregation among replicas.
The randomly elected leader (\secref{subsec:rang}) in a wave can be committed only if it belongs to $S^*$.
Hence, the probability of successfully committing a wave is $p = \frac{|S^*|}{N}$.
Given that each wave proceeds sequentially through $k$ rounds, the total expected latency (measured in the number of communication rounds) is:

\vspace{-3mm}
\[
\text{Expected Latency} = k \times \frac{1}{p} = \frac{k \times N}{|S^*|}.
\]
\vspace{-7mm}

\end{proof}

\subsubsection{Theoretical Analysis.}~\label{sec:bounds}
We now determine the number of rounds $k$ in \sysname that guarantees the common-core property (\ie, liveness) and yields low commit latency. Before that, we first introduce some notations for analysis. 
We use $R_1, \dots, R_k$ to denote the vertices in the rounds of a wave. 
Besides, for a vertex $x\in R_1$, we use $supp_k(x)$ to denote the number of round-$k$ vertices that are (causally) carried from $x$.
As shown in \apxref{app:bridge}, any lower bound on the number of first-round vertices satisfying $supp_k(x)\ge f{+}1$ in the DAG directly implies the same lower bound on $|S^\ast(k)|$, so we can equivalently reason in the DAG domain.

We next analyze feasibility by incrementally trying $k{=}2$---the wave shape adopted by early DAG protocols---to larger $k$.


\bheading{Liveness Violation for $k{=}2$.}
The appendix~\cite{fides_appendix} demonstrates that an asynchronous adversary can schedule a two-round wave so that no first-round vertex reaches $\mathrm{supp}_2(x)\!\ge\! f{+}1$, i.e., $|S^\ast(2)|{=}0$. 
Hence, the per-wave success probability $p_2{=}0$, and liveness cannot be guaranteed in adversarial networks.

\bheading{Liveness Exists, but Not Enough for $k{=}3$.}
The appendix~\cite{fides_appendix} shows that \textit{exactly two} first-round vertices reach the threshold in the worst case, so $|S^\ast(3)|{=}2$. We have
\[
p=\frac{2}{2f{+}1} \quad\Rightarrow\quad
\mathbb{E}[\textnormal{rounds}]=3\cdot\frac{2f{+}1}{2}=3f+\tfrac{3}{2}.
\]
Liveness holds, but the expected commit time grows linearly with $f$, so there is \textit{no constant-round guarantee}.
We thus have the theorem:
\begin{theorem}[Minimum Rounds ($k$)]\label{thm:min_rounds}
The $k$-iteration common core property holds when $k \ge 3$.
\end{theorem}

\begin{proof}[Sketch]
By the appendix proof~\cite{fides_appendix} and the DAG$\Rightarrow$Common-Core bridge, an adversary can realize $S^\ast(2)=\varnothing$; thus $k{=}2$ fails liveness. 
For $k{=}3$, exactly two first-round vertices achieve $\mathrm{supp}_3(\cdot)\!\ge\! f{+}1$, hence $|S^\ast(3)|{=}2>0$ via the bridge. 
Therefore, a non-empty common core exists if and only if $k\ge3$.
\end{proof}

\bheading{Optimal Constant-Time Threshold for $k{=}4$.}
With four rounds, our worst-case construction shows that \(|S^\ast(4)|\ge f{+}1\) always holds.
Combined with the upper bound in the appendix~\cite{fides_appendix}, this yields tightness: \(|S^\ast(4)|=f{+}1\) in the worst case.
Then
\[
p=\frac{f{+}1}{2f{+}1}\quad\Rightarrow\quad
\mathbb{E}[\textnormal{rounds}]=4\cdot\frac{2f{+}1}{f{+}1}\le 8,
\]
establishing a \textit{constant} worst-case expected commit latency.

\bheading{No Benefit in the Worst Case for $k{>}4$.}
An adversary can always suppress \(f\) proposers in round~1, forcing \(|R_1|=f{+}1\). Lemma~\apxlmref{lm:lge4-upper-fplus1}{app:subsec:liveness} caps the worst-case common core at \(|S^\ast(k)|\le f{+}1\), independent of \(k\). 
Hence, increasing the wave length beyond four cannot enlarge the worst-case core and only \textit{increases} the expected commit latency (Lemma~\apxlmref{lm:four-rounds-optimal}{app:subsec:liveness}), which yields the following theorem:

\begin{theorem}[Optimal Number of Rounds] \label{thm:optimal_k}
In \sysname, the optimal number of rounds per wave is $k = 4$.
\end{theorem}

\vspace{-4mm}
\begin{proof}[Sketch]
With a global common coin, $p= \frac{f{+}1}{2f{+}1}$, so $\mathbb{E}[\mathrm{rounds}]=4\cdot\frac{2f{+}1}{f{+}1}\le 8$ (Lemma~\apxlmref{lm:four_round_commit_exp}{app:subsec:liveness}), \ie, constant-round commitment. 
Further for any $k>4$, the same upper bound $|S^\ast(k)|\le f{+}1$ holds while $k$ increases, thus the worst-case expected latency strictly worsens (Lemma~\apxlmref{lm:four-rounds-optimal}{app:subsec:liveness}). 
Therefore, $k=4$ is optimal among constant-$k$ designs.
\end{proof}


\vspace{-3mm}
\bheading{Conclusion:} We summarize the choice of round number \(\boldsymbol{k}\) as:
\begin{itemize}[leftmargin=10pt]
    \item $k=2$: $|S^\ast| = 0$ — no liveness guarantee.
    \item $k=3$: $|S^\ast| = 2$ — ensures liveness, but commit may require $\Theta(f)$ rounds.
    \item $k=4$: $|S^\ast| = f{+}1$ — ensures liveness and constant-round commit.
    \item $k>4$: $|S^\ast| = f{+}1$ — additional rounds only increase latency without improving commit.
\end{itemize}
Therefore, we choose $k=4$ as the optimal rounds in a wave, achieving the best worst-case commit latency.

\subsection{Commit Rule under $k=4$ Regime} \label{sec:commitrule}
Having fixed $k=4$ under this quorum, we now instantiate the ordering layer. 
\sysname operates in four-round waves, and commitment is derived wave by wave directly from the DAG with no extra consensus messages---consistent with prior DAG-BFT designs. 
The commit rule process in \sysname, as outlined in \figref{algorithm:main}, begins with a DKG protocol \cite{Shamirsecret} at startup to set up the initial state of the TEE. 
A key improvement of \sysname is the utilization of the lightweight \textit{T-Coin} to generate globally consistent random values for leader election. 
This coin introduces negligible computational and communication overhead, ensuring efficient leader election.

Specifically by Theorem~\ref{thm:optimal_k}, we adopt a wave with four consecutive rounds to preserve liveness and optimal commit latency (as proven in \apxref{app:subsec:liveness}). 
Conceptually, during the first round, replicas propose vertices that encapsulate their entire causal history. 
The second and third rounds involve replicas voting on these proposals by referencing them in subsequent vertices to enhance the connectivity of the DAG. 
The fourth round uses the \textit{T-Coin} to retrospectively elect a leader vertex from the first round of the wave (line~\ref{line:rand}). 
A leader vertex is committed if it is referenced by at least $f{+}1$ vertices in the fourth round (line~\ref{f+1commit}). 
Once a leader vertex is committed, all its preceding causal history is transitively committed in a predefined order (line~\ref{line:predefined_order}, \eg, depth-first search).



Since replicas may have divergent DAG views, not all commit a leader each wave.
\sysname employs a recursive mechanism (line~\ref{line:recursivecheck}):
\begin{packedenumerate}
    \item Once a leader is committed in wave $w$, it becomes the candidate for transaction ordering.
    \item The system recursively checks the preceding waves to identify the most recent wave $w'$ with a committed leader.
    \item For waves between $w'$ and $w$, paths between the current candidate leader and the leaders of earlier waves are transitively committed in ascending order ($w'+1, w'+2, \ldots, w$).
\end{packedenumerate}
\figref{fig:commit} illustrates a recursive committing process in \sysname.
Consequently, an elected wave leader vertex $l$ of wave $w$ is committed by a non-faulty replica $p$ through either of the following mechanisms:
\begin{packedenumerate}
    \item \textbf{Direct Commit:} The leader $l$ is directly committed when, in the replica's local view, it is referenced by at least $f{+}1$ vertices from the fourth round of wave $w$.
    \item \textbf{Indirect Commit:} The leader $l$ is indirectly committed if there exists a path from $l_s$ to $l$, where $l_s$ is the directly committed leader vertex of a wave $w_s > w$.
\end{packedenumerate}



\sysname achieves better fault tolerance (\ie, tolerating $f$ faults in an $n{=}2f{+}1$ setting) and enhanced scalability, with safety and liveness formally proven in \apxref{app:correctproof}.
The appendix~\cite{fides_appendix} establishes that \sysname can expect to commit a leader vertex approximately every \textbf{four} rounds under random message delays, and every \textbf{eight} rounds under an asynchronous adversary.

\section{Correctness Analysis} \label{sec:correctness_analysis}
We here provide a sketch of \sysname's correctness, while deferring full proofs to \apxref{app:correctproof}.
Before that, we first prove the correctness of two building blocks, \ie, \textit{T-RBC} and \textit{T-Coin}.

\bheading{Correctness Analysis of Building Blocks.} 
\sysname utilizes two standard primitives: \textit{Reliable Broadcast} and \textit{Global Common Coin}.
We further instantiate them as TEE-assisted Reliable Broadcast \textit{(T-RBC)} and TEE-assisted Common Coin \textit{(T-Coin)}.
\begin{packedenumerate}
    \item For \textit{T-RBC}, we prove the RBC guarantees, \ie, Validity, Integrity, and Agreement, even under Byzantine faults~\cite{brachaRBC}. 
    
    \item For \textit{T-Coin}, we also prove the Global Common Coin properties, \ie, Agreement, Termination, and Unpredictability~\cite{DAGRider}.
    
\end{packedenumerate}
Full proofs of these building blocks appear in \apxref{app:build_block}.


\bheading{Safety Sketch Analysis of \sysname.} Safety follows from two consistency ingredients (\apxref{app:subsec:safety}).
The proof proceeds as follows:

\begin{packedenumerate}
    \item \textit{Unique leader per wave.} Any two correct replicas select the same leader for the same wave. This holds because, after converging to an identical DAG view (\textit{Reliable Dissemination}), correct replicas feed identical inputs and seeds into \textsc{RNG},  thereby deterministically electing the same leader (Lemma~\apxlmref{clm:uniquewaveleader}{app:subsec:safety}). 
    \item \textit{Canonical commit order.} All correct replicas commit leader vertices in the same ascending order of waves.  This is guaranteed because the algorithm commits leaders via a LIFO stack and pops them in increasing wave order (Lemma~\apxlmref{lm:sameorder}{app:subsec:safety}).
\end{packedenumerate}

Together, \sysname guarantees safety: replicas always agree on the elected leaders and commit them in a consistent order. Hence, the causal history of each leader is committed in the same order across all replicas, ensuring that the complete sequence of committed transactions is identical (Theorem~\apxlmref{thm:safety}{app:subsec:safety}).


\bheading{Liveness Analysis of \sysname.}
Liveness in \sysname follows from two key properties (with full arguments in \apxref{app:subsec:liveness}), bolstered by our earlier analysis in Sec.~\ref{sec:bound}.
The outline is as follows:

\begin{packedenumerate}
\item \textit{Non-zero per-wave commit probability via $k=4$.} Any wave always contains at least $f{+}1$ first-round vertices that meet the support threshold (Sec.~\ref{sec:bound}). 
Thus, each wave has a strictly positive probability of committing in both adversarial and random networks, which ensures progress in expectation.
\item \textit{Eventual inclusion of all vertices.} 
When a wave leader is committed, the protocol delivers its entire causal history in a deterministic order.
By utilizing weak edges and repeated commitment across waves, every vertex will be included in some committed leader’s causal history. 
Hence, every correct replica eventually commits every vertex (Lemma~\apxlmref{lm:vertex_commit}{app:subsec:liveness}).

\end{packedenumerate}
Together, these properties establish liveness: every wave progresses with a non-zero constant probability, and all vertices are eventually incorporated into the causal histories of committed leaders to ensure their final commitment (Theorem~\apxlmref{thm:Liveness}{app:subsec:liveness}).

\section{Evaluation} \label{sec:evaluation}
We compare the performance of \sysname with several state-of-the-art BFT consensus protocols: Tusk~\cite{narwhal} (an async DAG protocol), Bullshark~\cite{Bullshark} (an async DAG protocol with a partial sync fast path), Mysticeti~\cite{babel2024mysticeti} (a partial sync DAG BFT protocol), \new{Shoal++~\cite{shoal++} (a partial sync DAG protocol with pipelined anchors),} RCC~\cite{gupta2021rcc} (a multi-leader BFT protocol), Damysus~\cite{Damysus} (a TEE-assisted BFT protocol), Achilles~\cite{Achilles} (a TEE-assisted rollback-resilient protocol) and HybridSet~\cite{TrustedHardwareAssisted} (a TEE-assisted leaderless BFT protocol).
We measure throughput and latency under varying conditions (\eg, faulty nodes, network latency, batching). 
Our experiments aim to answer the following questions:
\begin{packeditemize}
    \item \textbf{Q1:} How does \sysname scale with larger $f$, batch sizes, and workload compared to prior DAG protocols? (\secref{subsec:scalability})

    \item \textbf{Q2:} How well does \sysname sustain performance under injected delays and leader failures? (\secref{subsec:resiliency})

    \item \new{\textbf{Q3:} How does \sysname behave under adversarial conditions, including crash faults, Byzantine attacks, and asynchrony? (\secref{subsec:resiliency})}

    \item \textbf{Q4:} What is the breakdown of \sysname's end-to-end latency, and how do individual and combined TEE components contribute to the overall performance? (\secref{subsec:breakdown})
    
\end{packeditemize}

\begin{figure*}[t]
    \vspace{-4.5mm}
    \centering
    \includegraphics[width=1\linewidth]{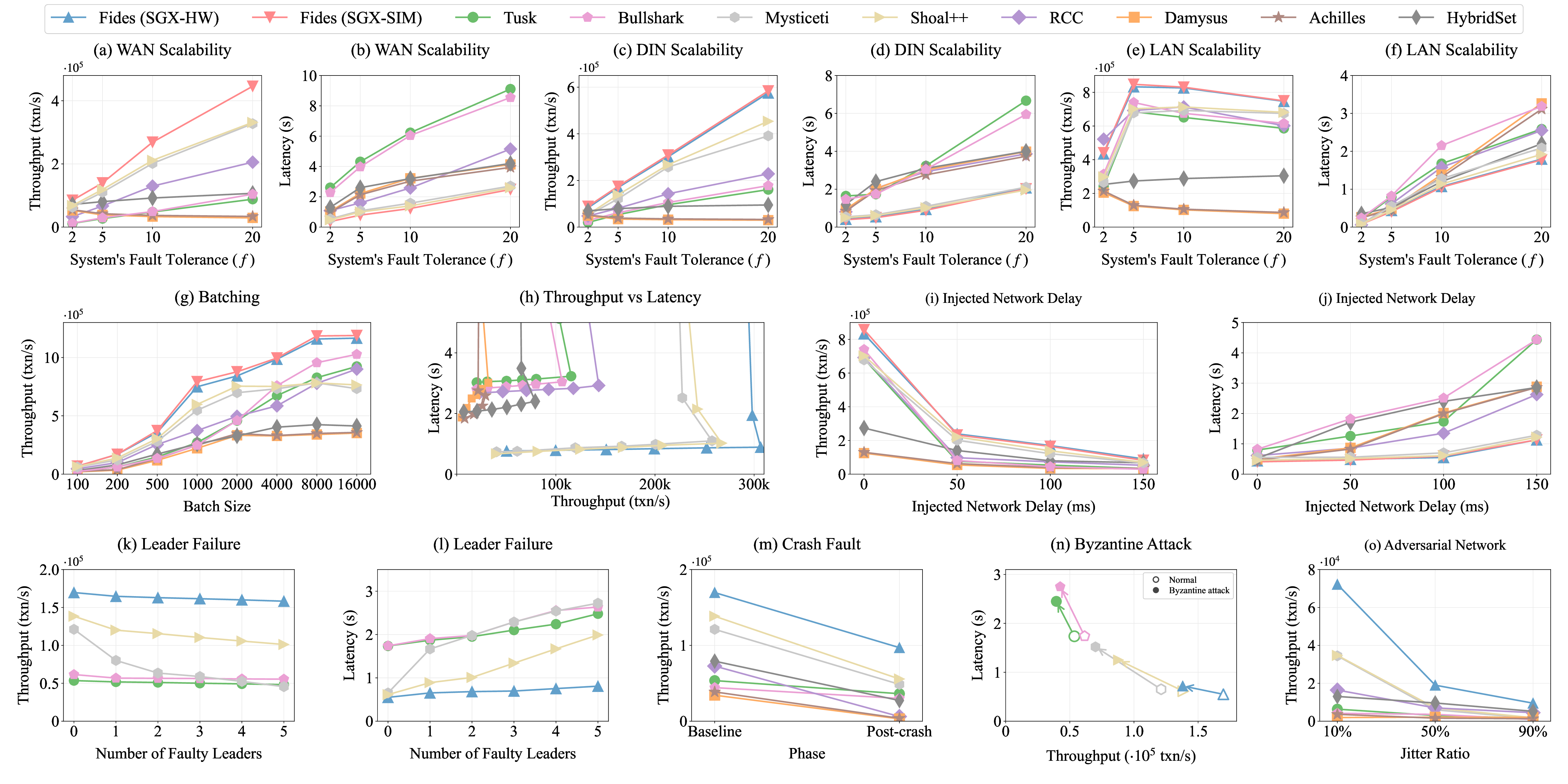}
    \vspace{-9mm}
    \captionsetup{labelfont={bf}}
    \caption{\new{Results for Scalability and Resiliency Evaluation.}}
    \vspace{-5mm}
    \label{fig:scalability}
\end{figure*}

\subsection{Implementation and Evaluation Setup}
\bheading{Implementation.} We built \sysname atop the open-source Apache ResilientDB platform~\cite{resdb, resdbdemopaper, resdbpaper}, a scalable and global blockchain infrastructure. 
Our implementation is developed in C++~\cite{sourcecode} and leverages Apache ResilientDB for durable data storage.
We use TCP to implement reliable point-to-point communication. 
The core trusted components (\ie, \textit{MC}, \textit{RAC}, and \textit{RNG}) run within Intel SGX enclaves~\cite{intelsgx} in each replica, integrated via the Open Enclave SDK~\cite{openenclave2022}.
\new{All protocols share ResilientDB's networking, storage, and cryptographic stack, differing only in the consensus layer. Reported measurements include all TEE-related overheads if applicable.}

\bheading{Evaluation Setup.}
We conducted experiments on Alibaba Cloud Elastic Compute Service~\cite{alibaba-confidential-computing}. 
SGX-enabled runs used \texttt{ecs.g7t.xlarge} instances in the Hong Kong region, which support Intel SGX; baseline (non-SGX) runs used \texttt{ecs.g7.xlarge} instances. 
Both instance types provide 4 vCPUs and 16 GB of memory, and all VMs ran Ubuntu 22.04.
Each process ran on a dedicated VM configured with 4 vCPUs and 16 GB RAM, running Ubuntu 22.04.

In the LAN environment, all machines were deployed within the same data center in Hong Kong. 
For WAN experiments, due to the limited availability of SGX-enabled instances, we adopt two complementary settings:
\begin{packeditemize} 
\item \textbf{Wide Area (Geo-Distributed) Network (WAN):} We deploy replicas across four data centers located in the United States, UAE, Singapore, and Germany, with inter-datacenter latencies ranging from 80 ms to 270 ms. 
Each replica runs trusted components in \textit{simulated SGX} mode (\textit{SGX-SIM})~\cite{sgxsimulate}\new{, as SGX-enabled instances are unavailable in some regions; we justify this in \apxref{sec:sgx-overhead}.}

\item \textbf{Delay-Injected Network (DIN):} We enable \textit{SGX Hardware encryption} (\textit{SGX-HW}) but emulate WAN conditions within a region by injecting specific end-to-end delays between machines. A default delay of 100 ms is used to represent a typical WAN RTT. 
\end{packeditemize}
The WAN deployment enables us to evaluate the impact of real network delays in wide-area settings. Meanwhile, the DIN environment allows for the measurement of the actual performance of SGX under WAN-like latency conditions.

We performed experiments in three network environments: LAN, WAN, and DIN. 
In each setup, each node exclusively ran one instance of \sysname, 
with the number of clients matching the number of consensus processes to ensure that replicas were fully saturated.
Clients generated a saturating request load for replicas.  
The benchmark was a key-value service, where each request involved a randomly generated key-value insertion with a fixed size of 50 B.
We performed three runs per condition and plotted the median value.

\subsection{Scalability Evaluation} \label{subsec:scalability}

\subsubsection{Varying Fault Tolerance.} \label{subsubsec:faulttolerance} \label{evaluation:fault_tolerance}
We evaluate performance across varying fault-tolerance levels. 

\bheading{Performance under WAN.} In WAN environments, where network latencies are higher, scalability is a critical challenge for consensus protocols. 
\figref{fig:scalability}(a) and (b) depict throughput and latency as the system's fault tolerance $f$ increases from 2 to 20 in the WAN setting.
Across the entire range, \sysname demonstrates superior scalability:
Its throughput remains highest across all $f$, exceeding Tusk by over $4\times$ and doubling RCC at higher $f$.
In contrast, other protocols either saturate early or suffer sharply rising latency as faults increase.
\new{Shoal++ and Mysticeti are the closest DAG competitors, yet at $f{=}20$ \sysname still outperforms both by $\sim$35\%, owing to its reduced quorum size ($f{+}1$ out of $2f{+}1$ vs.\ $2f{+}1$ out of $3f{+}1$) and lower message complexity.}
Damysus and Achilles exhibit significantly lower throughput because their leader-centric communication patterns suffer especially in cross-region scenarios.
HybridSet performs slightly better than the other two but still lags behind \sysname, as it relies on a hybrid binary agreement to decide on each proposal.


\bheading{Performance under DIN.}
To evaluate the overhead introduced by TEE, particularly in \textit{SGX-HW} mode, we conducted tests in DIN environments. 
\figref{fig:scalability}(c) and (d) illustrate the scalability results in this scenario.
In DIN environments, \sysname-HW and \sysname-SIM perform almost identically.
The negligible performance gap between the HW and SIM modes demonstrates that the real hardware-based encryption in SGX-HW mode introduces minimal computational overhead.
Both of them significantly outperform all other protocols---consistent with results shown under WAN.


\bheading{Performance under LAN.} \figref{fig:scalability}(e) and (f) illustrate the results in LAN, where network latencies are negligible.
Even at $f{=}20$, \sysname maintains $\sim$746k txn/s throughput with $\sim$1.80 s latency, while competing protocols fall to \new{80k--681k} txn/s and \new{1.9--3.2} s latency. 
Although \sysname's RBC optimization yields less dramatic gains in LAN, it still scales much more gracefully than alternatives.





\subsubsection{Varying Batch Size.} \label{subsubsec:batch}
\figref{fig:scalability}(g) illustrates the impact of batch size on performance with a fault tolerance of $f{=}5$.
All protocols improve with larger batches before saturating, but \sysname consistently leads: 
at batch size = 4,000, \sysname achieves $\sim$985k txn/s, while others lie in the range \new{328k--758k} txn/s. 
As batch size grows further, throughput growth slows, yet \sysname retains its lead consistently.

\subsubsection{Varying Workload.}
\figref{fig:scalability}(h) varies workload at $f{=}10$.
All protocols improve until saturation, with \sysname consistently leading.
For instance, near saturation, \sysname reaches $\sim$300k~txn/s at $\sim$0.9~s latency, while the next-best \new{competitors (Shoal++ and Mysticeti) achieve $\sim$267k and $\sim$258k txn/s} but with latency above 1.0 s, and others drop below 150k txn/s.

\subsection{Resiliency Evaluation} \label{subsec:resiliency}
We evaluate \sysname's resiliency by examining \new{the impact of network latencies, leader failures, and adversarial network conditions}.

\subsubsection{Injecting Network Delay.} \label{subsubsec:delay}
We evaluate the performance of all protocols under network delays injected from 0 to 150 ms.
\figref{fig:scalability}(i,j) shows \sysname (HW and SIM) maintains higher throughput and lower latency than all baselines across every delay setting.

\subsubsection{Varying Leader Failures.} \label{subsubsec:failure}
We evaluated \new{five} DAG-based protocols (\sysname, Tusk, Bullshark, Mysticeti\new{, Shoal++}) under $f{=}5$ in the DIN environment.
We focus on these \new{five} protocols because they share a comparable DAG architecture, enabling a consistent and fair failure definition across systems.
Here, a failed leader is defined as a replica whose first-round vertex never gathers enough support (\ie, never reaches the commit quorum), and hence is never committed.


\figref{fig:scalability}(k) shows that \sysname delivers the best performance, with gradual declines for all as the number of failures increases from 0 to 5, whereas Tusk, Bullshark, \new{Shoal++,} and especially Mysticeti drop off more sharply.
In terms of latency (\figref{fig:scalability}(l)), although Mysticeti performs well in the good case, it degrades quickly once failures occur.
Meanwhile, \sysname maintains the lowest latency and exhibits more moderate growth compared to others, showcasing its robustness and resiliency in handling node failures.

\subsubsection{\new{Dynamic Crash Faults.}} \label{subsec:adversarial}
\new{Since \sysname's TEE-assisted design rules out equivocation and forgery, we focus on omission, crash faults, and network asynchrony. 
In this setting, we crash $f$ replicas after 10\,s. 
As shown in \figref{fig:scalability}(m), all protocols experience performance degradation after the crashes, but \sysname remains the best, sustaining the highest post-crash throughput of $\sim$97k txn/s.
\sysname benefits from its lightweight design: TEE-assisted dissemination and a four-round commit path reduce coordination and communication overhead under failures.
}

\subsubsection{\new{Byzantine Delivery Attack.}}
\new{We further add a DAG-specific Byzantine attack: up to $f$ Byzantine replicas send each vertex to only $f{+}1$ replicas and try to suppress leader-supporting references.
This creates worst-case local DAG views without equivocation or forgery.
\figref{fig:scalability}(n) shows \sysname remains the best protocol, sustaining 138.0k txn/s with 0.72\,s latency under the attack.
\sysname's advantage comes from its TEE-assisted dissemination and four-round commit rule, which improve DAG connectivity and reduce message overhead even when Byzantine replicas suppress references.}

\subsubsection{\new{Adversarial Delay.}}
\new{We fix the mean network delay at 500\,ms and increase the jitter ratio from 10\% to 90\% to emulate adversarially induced network instability.
\figref{fig:scalability}(o) shows that timeout-sensitive baselines such as Mysticeti, Shoal++, and Bullshark degrade sharply, falling to around 2k txn/s at 90\% jitter, whereas \sysname's fully asynchronous design maintains the highest throughput ($\sim$10k txn/s).}

\begin{figure*}[t!]
    \vspace{-4mm}
    \centering
    \includegraphics[width=1\linewidth]{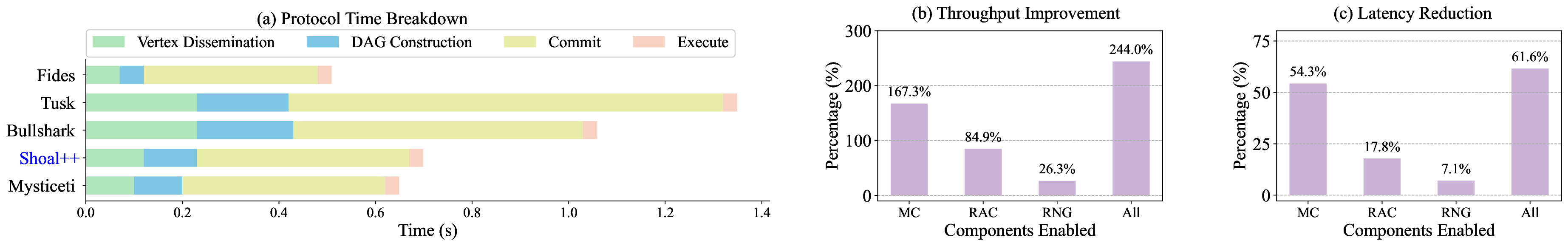}
    \vspace{-8mm}
    \captionsetup{labelfont={bf}}
    \caption{\new{System Breakdown by Protocol Time and Trusted Components.}}
    \vspace{-4mm}
    \label{fig:breakdown}
\end{figure*}

\subsection{System Breakdown} \label{subsec:breakdown}
\subsubsection{Latency Breakdown.} 
We provide a breakdown of the \new{five} DAG-based protocols in a DIN environment with $f{=}5$. 
For fair comparison, transaction queuing latency is excluded to mitigate the impact of clients. 
The latency is divided into four stages:  
\begin{packeditemize}
\item \textbf{Vertex Dissemination.} 
\sysname disseminates a vertex in 0.07 s, outperforming Tusk (0.23 s) and Bullshark (0.23 s) while remaining comparable to Mysticeti (0.1 s), enabled by its optimized T-RBC.

\item \textbf{DAG Construction.} 
From confirmation to local insertion into DAG, \sysname takes 0.05 s---a 75\% reduction versus Tusk and Bullshark---by streamlining the vertex validation.

\item \textbf{Commitment.} \sysname achieves a commit time of 0.36 s, approximately 50\% faster than other protocols.

\item \textbf{Execution.} All protocols have similar execution times of about 0.03 s, as no optimizations on execution are applied.
\end{packeditemize}

\figref{fig:breakdown}(a) shows the average time spent in each stage. The breakdown results indicate that \sysname achieves a minimum end-to-end latency of 0.5 s, driven primarily by its accelerated T-RBC-assisted dissemination, DAG construction, and commitment.
\new{Notably, \sysname achieves comparable latency to Mysticeti and Shoal++ despite using a certified DAG.
These approaches converge to similar end-to-end latency through different mechanisms: Mysticeti and Shoal++ remove per-vertex certification (but require partial synchrony for liveness), while \sysname makes certification cheap by replacing multi-round RBC with T-RBC's single broadcast and using RAC to reduce certificate validation overhead.}

\subsubsection{Trusted Component Breakdown.}\label{subsec:component_breakdown}  
To analyze the performance improvement brought by each trusted component, we ran four configurations: three cases with exactly one component enabled \textit{(MC, RNG, RAC)}, and an \textit{All} configuration with all components active under the $f{=}10$ setting in DIN.
We measured throughput and latency for each configuration and then calculated the relative improvement compared to the baseline (\ie, \sysname but with no trusted component active). 
In all configurations, \sysname preserves safety and liveness.
Disabling a trusted component only affects performance because the protocol falls back to corresponding untrusted compensations (e.g., extra communication or heavier cryptographic operations to tolerate Byzantine behavior).

\figref{fig:breakdown}(b) and (c) summarize these changes.
Enabling the \textit{MC} component delivers the largest improvement, increasing throughput by 167.3\% while reducing latency by 54.3\%.  
\textit{RAC} and \textit{RNG} deliver additional throughput improvements of 84.9\% and 26.3\%, and latency reductions of 17.8\% and 7.1\%, respectively.
When all three components are enabled simultaneously, their effects combine synergistically: throughput increases by 244.0\% and latency decreases by 61.6\%.
This \textit{All} configuration demonstrates the substantial performance gains achievable through our trusted-component architecture, while still preserving transaction confidentiality.

\section{Discussion} \label{sec:discussion}

\noindent\new{\textbf{Hybrid fault model and TEE availability.} \sysname{} can be extended to a hybrid fault model with $n{=}2f{+}2c{+}1$ and quorum $f{+}c{+}1$, tolerating $c$ honest replicas with \textit{crashed TEEs}.
Under systemic TEE compromise, \sysname{} falls back to the classical $n{=}3f{+}1$ with quorum $2f{+}1$, analogous to FlexiTrust~\cite{FlexiTrust}.
Beyond fault counts, MC and RAC require only TEE \textit{integrity}~\cite{trInc, hybster, minbft}, while RNG-seed leakage breaches \textit{confidentiality}, weakening liveness but not safety.}

\noindent\new{\textbf{Deployment considerations.} 
\sysname{} is not tied to Intel SGX: its trusted components form an abstract interface with three primitives realizable on any TEE, including enclave-style platforms (Intel SGX/TDX) and VM-based platforms such as AWS Nitro Enclaves~\cite{aws-nitro-enclaves}.
While TEEs introduce operational complexity and vendor dependence, these are acceptable in \sysname{}'s target setting of consortium systems, where participants can agree on hardware requirements and enforce policies such as remote attestation, mandatory microcode updates, and timely patching. 
TEE availability is also increasingly practical, with major cloud providers such as Alibaba Cloud, Azure, and AWS offering TEE-enabled instances~\cite{alibaba-confidential-computing, azure-confidential-computing, aws-nitro-enclaves}.
}

\section{Related Work}
We survey prior work on concurrent protocols, DAG-based BFT consensus and hardware-assisted BFT protocols, with a particular focus on TEE-assisted designs. While a growing body of work explores TEE-assisted BFT, very few efforts have investigated how TEEs can be effectively integrated into DAG-based BFT consensus.

\bheading{Concurrent / multi-leader BFT protocols.}
Many classic and streamlined BFT protocols are leader-driven: in each view (round), a designated leader proposes a block and replicas vote to commit it.
Despite many recent advances~\cite{hotstuff-2, hotstuff-1, carry}, leader-based protocols can still suffer from a \textit{leader bottleneck}, where a single leader must disseminate proposals and aggregate votes.
As the replica set grows, this workload can limit throughput and increase latency.
Concurrent (multi-leader) designs mitigate this bottleneck by enabling multiple leaders (or multiple instances) to make progress in parallel.
For example, Mir-BFT~\cite{mirbft} supports parallel leaders that independently propose request batches.
Later multi-BFT systems further multiplex multiple leader-based consensus instances in parallel~\cite{gupta2021rcc, ISS, spotless, Ladon2025, dqbft}, but typically require an additional mechanism (e.g., a global ordering layer) to reconcile outputs into a single ledger.
Orthrus~\cite{Orthrus2025} and Hydra~\cite{hydra} further reduce confirmation latency by exploiting transaction-level concurrency, \eg, by relaxing strict serialization via partial ordering or by removing the global ordering requirement when possible.
More recently, DAG-based BFT protocols offer a complementary path to scalability by allowing many replicas to disseminate and reference batches concurrently.
\sysname exemplifies this direction as a DAG-based BFT protocol that alleviates the single-leader bottleneck via DAG concurrency.

\bheading{DAG-based BFT protocols.} 
Recently, DAG-based BFT protocols~\cite{DAGRider, narwhal, Bullshark, shoal, Hashgraph, shoal++, babel2024mysticeti, jovanovic2024mahi, malkhi2023bbca, Fino, autobahn, gradeddag, LightDAG, sailfish} 
link transactions into a DAG, enabling greater scalability and higher throughput compared to linear-chain approaches.
HashGraph~\cite{Hashgraph} is among the first protocols to incorporate an asynchronous consensus mechanism into a DAG, which separates the network communication layer from the ordering logic.
DAG-Rider~\cite{DAGRider} advances asynchronous BFT consensus by combining reliable broadcast and a round-by-round DAG structure, which optimizes resilience, round complexity, and amortized communication complexity. 
Narwhal and Tusk~\cite{narwhal} introduce a modular approach that decouples data availability from transaction ordering, while Bullshark~\cite{Bullshark} adapts Narwhal/Tusk to improve performance in partially synchronous environments. 
In contrast, partially synchronous DAG-based protocols, such as \cite{shoal, shoal++, malkhi2023bbca, babel2024mysticeti, sailfish}, achieve low latency in well-behaved networks but do not fully address challenges under asynchronous conditions.
\new{Recent variants further explore this design space: Shoal++~\cite{shoal++} pipelines anchors and aggregates votes to lower commit latency, Sailfish~\cite{sailfish} optimizes DAG commit latency, GradedDAG~\cite{gradeddag} adopts graded reliable broadcast to streamline DAG construction, and Mahi-Mahi~\cite{jovanovic2024mahi} extends Mysticeti's uncertified DAG with an asynchronous global coin to achieve asynchronous liveness.}
\sysname builds on these asynchronous DAG protocols and leverages TEEs to further improve scalability in adversarial or unpredictable network settings.

\bheading{Hardware-assisted BFT protocols.}
Hardware-assisted BFT protocols~\cite{trInc, hybster, Damysus, FlexiTrust, TeeRollup, OneShot, TrustedHardwareAssisted} leverage trusted hardware to improve the performance and security of BFT consensus. TrInc~\cite{trInc} introduces a non-decreasing counter to prevent equivocation in distributed systems, reducing the fault tolerance requirement from $n{=}3f{+}1$ to $n{=}2f{+}1$, and simplifying the consensus process from three phases to two. 
Hybster~\cite{hybster} combines traditional Byzantine fault tolerance with a trusted crash-only subsystem, achieving over 1 million operations per second in an Intel SGX-based prototype. 

Damysus~\cite{Damysus} employs the \textit{Checker} and \textit{Accumulator} trusted components to streamline BFT and reduce communication phases.
In cases where the trustworthiness of replicas and their TEEs is uncertain, FlexiTrust~\cite{FlexiTrust} adapts the replication threshold between $2f{+}1$ and $3f{+}1$. 
OneShot~\cite{OneShot} presents a view-adapting BFT protocol that uses TEEs for efficient fault tolerance, while Achilles~\cite{Achilles} adopts rollback-resilient recovery and chained commit rules to achieve linear message complexity and low end-to-end transaction latency. 
Trusted Hardware-Assisted Leaderless BFT~\cite{TrustedHardwareAssisted} (referred to as HybridSet in the experiments) eliminates the need for a leader, mitigating leader-induced vulnerabilities. 
In contrast, \sysname goes a step further by leveraging TEE's guarantees to offload expensive tasks in DAG-based protocols into enclaves, thus achieving significant performance improvements.

\section{Conclusion} \label{sec:conclusion}
We presented \sysname, among the first asynchronous DAG-based BFT consensus protocols that integrate TEEs to tolerate a minority of Byzantine replicas. 
\sysname adopts a customized TEE-assisted \textit{T-RBC} and reconstructs the DAG by relaxing the vertex references from $2f{+}1$ to $f{+}1$.
More importantly, this reduced quorum introduces a liveness pitfall for standard two-round DAG commit rules under adversarial asynchrony; \sysname resolves it with a four-round commit rule that restores liveness and achieves optimal constant expected commit latency.
Together, the three trusted components in \sysname minimize network communication complexity and reduce reliance on cryptography.
Comprehensive evaluations demonstrate that \sysname consistently outperforms state-of-the-art protocols in both throughput and latency, validating its practicality for real-world systems.

\normalem
\begin{acks}
This work is partially funded by NSF Award Number 2245373.
\end{acks}

\bibliographystyle{unsrt}
\bibliography{bib}

@inproceedings{Ladon2025,
  title={{Ladon: High-Performance Multi-BFT Consensus via Dynamic Global Ordering}},
  author={Lyu, Hanzheng and Xie, Shaokang and Niu, Jianyu and Feng, Chen and Zhang, Yinqian  and Beschastnikh, Ivan},
  booktitle={EuroSys},
  year={2025}
}

@inproceedings{Orthrus2025,
  title={{Orthrus: Accelerating Multi-BFT Consensus through Concurrent Partial Ordering of Transactions}},
  author={Lyu, Hanzheng and Xie, Shaokang and Niu, Jianyu and Sadoghi, Mohammad and Beschastnikh, Ivan and Zhang, Yinqian and Feng, Chen},
  booktitle={IEEE ICDE},
  year={2025}
}

@inproceedings{Achilles,
author = {Niu, Jianyu and Wen, Xiaoqing and Wu, Guanlong and Liu, Shengqi and Yu, Jiangshan and Zhang, Yinqian},
title = {{Achilles: Efficient TEE-Assisted BFT Consensus via Rollback Resilient Recovery}},
year = {2025},
booktitle = {EuroSys},
}

@inproceedings{narrator,
author = {Niu, Jianyu and Peng, Wei and Zhang, Xiaokuan and Zhang, Yinqian},
title = {{NARRATOR: Secure and Practical State Continuity for Trusted Execution in the Cloud}},
year = {2022},
booktitle = {ACM CCS}
}

@inproceedings{narrator-pro,
  author={Peng, Wei and Li, Xiang and Niu, Jianyu and Zhang, Xiaokuan and Zhang, Yinqian},
  booktitle={IEEE TDSC}, 
  title={Ensuring State Continuity for Confidential Computing: A Blockchain-based Approach}, 
  year={2024},
}

@inproceedings{fairness,
author = {Abraham, Ittai and Malkhi, Dahlia and Spiegelman, Alexander},
title = {{Asymptotically Optimal Validated Asynchronous Byzantine Agreement}},
year = {2019},
booktitle = {ACM PODC}
}

@inproceedings{BLS2001,
author="Boneh, Dan
and Lynn, Ben
and Shacham, Hovav",
title="Short Signatures from the Weil Pairing",
booktitle="ASIACRYPT",
year="2001"
}

@misc{waves,
  author = {}, 
  year =         "2019",
  title =        "Waves platform",
  note =         "",
  url =          {https://docs.wavesplatform.com/en/blockchain/waves-protocol/waves-ng-protocol.html},
  howpublished = "",
  month =        "Nov",
  lastaccessed = "",
}

@inproceedings{pbft1999,
 author = {Castro, Miguel and Liskov, Barbara},
  title={{Practical Byzantine Fault Tolerance}},
  booktitle={USENIX OSDI},
  year={1999}
}

@article{tendermint,
  author    = {Ethan Buchman and Jae Kwon and Zarko Milosevic},
  title     = {The latest gossip on {BFT} consensus},
  journal   = {CoRR},
  volume    = {abs/1807.04938},
  year      = {2018},
  archivePrefix = {arXiv}
}

@inproceedings{avarikioti2020fnf,
  title={{FnF-BFT: Exploring Performance Limits of BFT Protocols}},
  author={Avarikioti, Zeta and Heimbach, Lioba and Schmid, Roland and Vanbever, Laurent and Wattenhofer, Roger and Wintermeyer, Patrick},
  booktitle={arXiv preprint arXiv:2009.02235},
  year={2020}
}

@inproceedings{ISS,
  title={{State-Machine Replication Scalability Made Simple}},
  author={Stathakopoulou, Chrysoula and Pavlovic, Matej and Vukoli{\'c}, Marko},
  booktitle={EuroSys},
  year={2022}
}

@inproceedings{gupta2021rcc,
  title={{RCC: Resilient Concurrent Consensus for High-Throughput Secure Transaction Processing}},
  author={Gupta, Suyash and Hellings, Jelle and Sadoghi, Mohammad},
  booktitle={IEEE ICDE},
  year={2021},
}

@inproceedings{mirbft,
   author={Stathakopoulou, Chrysoula and David, Tudor and Pavlovic, Matej and Vukolic, Marko},
   title = {Mir-BFT: Scalable and Robust {BFT} for Decentralized Networks},
   booktitle = {JSys},
   year = {2022},
}

@inproceedings{dqbft,
author = {Arun, Balaji and Ravindran, Binoy},
title = {{Scalable Byzantine Fault Tolerance via Partial Decentralization}},
year = {2022},
booktitle = {PVLDB}
}

@inproceedings{Dumbo-MVBA,
      author = {Yuan Lu and Zhenliang Lu and Qiang Tang and Guiling Wang},
      title = {{Dumbo-MVBA: Optimal Multi-valued Validated Asynchronous Byzantine Agreement, Revisited}},
      year = {2020},
      booktitle = {ACM PODC}
}

@inproceedings{DAGRider,
author = {Keidar, Idit and Kokoris-Kogias, Eleftherios and Naor, Oded and Spiegelman, Alexander},
title = {{All You Need is DAG}},
year = {2021},
booktitle = {ACM PODC},
}

@inproceedings{Bullshark,
author = {Spiegelman, Alexander and Giridharan, Neil and Sonnino, Alberto and Kokoris-Kogias, Lefteris},
title = {{Bullshark: DAG BFT Protocols Made Practical}},
year = {2022},
booktitle = {ACM CCS}
}

@inproceedings{trInc,
author = {Levin, Dave and Douceur, John R. and Lorch, Jacob R. and Moscibroda, Thomas},
title = {{TrInc: Small Trusted Hardware for Large Distributed Systems}},
year = {2009},
booktitle = {USENIX NSDI},
}

@inproceedings{hybster,
author = {Behl, Johannes and Distler, Tobias and Kapitza, R\"{u}diger},
title = {{Hybrids on Steroids: SGX-Based High Performance BFT}},
year = {2017},
booktitle = {EuroSys}
}

@inproceedings{Damysus,
author = {Decouchant, J\'{e}r\'{e}mie and Kozhaya, David and Rahli, Vincent and Yu, Jiangshan},
title = {{Damysus: Streamlined BFT Consensus Leveraging Trusted Components}},
year = {2022},
booktitle = {EuroSys}
}

@inproceedings{narwhal,
  title={{Narwhal and Tusk: A {DAG}-based Mempool and Efficient {BFT} Consensus}},
  author={Danezis, George and Kokoris-Kogias, Lefteris and Sonnino, Alberto and Spiegelman, Alexander},
  booktitle={EuroSys},
  year={2022}
}

@inproceedings{shoal,
  title={{Shoal: Improving DAG-BFT Latency and Robustness}},
  author={Spiegelman, Alexander and Arun, Balaji and Gelashvili, Rati and Li, Zekun},
  booktitle={FC},
  year={2024}
}

@inproceedings{malkhi2023bbca,
  title={{BBCA-CHAIN: Low Latency, High Throughput BFT Consensus on a DAG}},
  author={Malkhi, Dahlia and Stathakopoulou, Chrysoula and Yin, Maofan},
  booktitle={FC},
  year={2024}
}

@inproceedings{babel2024mysticeti,
      title={{Mysticeti: Reaching the Limits of Latency with Uncertified DAGs}}, 
      author={Kushal Babel and Andrey Chursin and George Danezis and Anastasios Kichidis and Lefteris Kokoris-Kogias and Arun Koshy and Alberto Sonnino and Mingwei Tian},
      year={2024},
      booktitle={arXiv preprint arXiv:2310.14821},
}

@inproceedings{minbft,
author = {Veronese, Giuliana and Correia, Miguel and Bessani, Alysson and Lung, Lau and Veríssimo, Paulo},
year = {2013},
title = {{Efficient Byzantine Fault-Tolerance}},
booktitle = {IEEE TC}
}

@inproceedings{intelsgx,
      author = {Victor Costan and Srinivas Devadas},
      title = {{Intel SGX Explained}},
      booktitle = {Cryptology ePrint Archive, Paper 2016/086},
      year = {2016},
}

@misc{sourcecode,
    title = {{Fides Source Code}},
    howpublished = {\url{https://github.com/apache/incubator-resilientdb/tree/fides}},
    note = {(Accessed: 04-29-2026)}
}

@misc{intelRA,
    title = {{Intel. {SGX} Remote Attestation}},
    howpublished = {\url{https://www.intel.com/content/www/us/en/developer/tools/software-guard-extensions/attestation-services.html}},
    note = {(Accessed: 04-29-2026)}
}

@misc{sgxsealing,
    title = {{Intel. {SGX} Sealing}},
    howpublished = {\url{https://www.intel.com/content/www/us/en/developer/articles/technical/introduction-to-intel-sgx-sealing.html}},
    note = {(Accessed: 04-29-2026)}
}

@misc{amdRA,
    title = {{AMD. SEV-SNP Platform Attestation}},
    howpublished = {\url{https://docs.amd.com/v/u/en-US/58217_amd-epyc-9004-ug-platform-attestation-using-virtee-snp}},
    note = {(Accessed: 04-29-2026)}
}

@misc{openenclave2022,
  title = {{Open Enclave SDK}},
  howpublished = {\url{https://openenclave.io}},
  note = {(Accessed: 04-29-2026)}
}

@misc{resdb,
  title = {{Apache ResilientDB}},
  howpublished = {\url{https://resilientdb.incubator.apache.org/}},
  note = {(Accessed: 04-29-2026)}
}

@inproceedings{Dumbo-NG,
author = {Gao, Yingzi and Lu, Yuan and Lu, Zhenliang and Tang, Qiang and Xu, Jing and Zhang, Zhenfeng},
title = {{Dumbo-NG: Fast Asynchronous BFT Consensus with Throughput-Oblivious Latency}},
year = {2022},
booktitle = {ACM CCS}
}

@inproceedings{intelsgxpaper,
author = {Hoekstra, Matthew and Lal, Reshma and Pappachan, Pradeep and Phegade, Vinay and Del Cuvillo, Juan},
title = {Using innovative instructions to create trustworthy software solutions},
year = {2013},
booktitle = {HASP}
}

@inproceedings {ROTE,
author = {Sinisa Matetic and Mansoor Ahmed and Kari Kostiainen and Aritra Dhar and David Sommer and Arthur Gervais and Ari Juels and Srdjan Capkun},
title = {{ROTE: Rollback Protection for Trusted Execution}},
booktitle = {USENIX Security},
year = {2017},
}

@inproceedings{shoal++,
      title={{Shoal++: High Throughput DAG BFT Can Be Fast!}}, 
      author={Balaji Arun and Zekun Li and Florian Suri-Payer and Sourav Das and Alexander Spiegelman},
      year={2024},
      booktitle={arXiv preprint arXiv:2405.20488}
}

@inproceedings{sailfish,
      author = {Nibesh Shrestha and Rohan Shrothrium and Aniket Kate and Kartik Nayak},
      title = {Sailfish: Towards Improving the Latency of {DAG}-based {BFT}},
      booktitle = {IEEE S\&P},
      year = {2025}
}

@inproceedings{Fino,
      title={{Maximal Extractable Value (MEV) Protection on a DAG}},
      author={Dahlia Malkhi and Pawel Szalachowski},
      year={2022},
      booktitle={Tokenomics}
}

@inproceedings{gradeddag,
      author = {Xiaohai Dai and Zhaonan Zhang and Jiang Xiao and Jingtao Yue and Xia Xie and Hai Jin},
      title = {{GradedDAG: An Asynchronous DAG-based BFT Consensus with Lower Latency}},
      booktitle = {IEEE SRDS},
      year = {2023}
}

@inproceedings{LightDAG,
      author = {Xiaohai Dai and Guanxiong Wang and Jiang Xiao and Zhengxuan Guo and Rui Hao and Xia Xie and Hai Jin},
      title = {{LightDAG: A Low-latency DAG-based BFT Consensus through Lightweight Broadcast}},
      booktitle = {IEEE IPDPS},
      year = {2024}
}

@inproceedings{autobahn,
      title={Autobahn: Seamless high speed BFT}, 
      author={Neil Giridharan and Florian Suri-Payer and Ittai Abraham and Lorenzo Alvisi and Natacha Crooks},
      year={2024},
      booktitle={ACM SOSP}
}

@inproceedings{Hashgraph,
  title={{The Hashgraph Protocol: Efficient Asynchronous BFT for High-Throughput Distributed Ledgers}},
  author={Leemon Baird and Atul Luykx},
  year={2020},
  booktitle={COINS}
}

@INPROCEEDINGS {OneShot,
author = { Decouchant, Jérémie and Kozhaya, David and Rahli, Vincent and Yu, Jiangshan },
title = {{OneShot: View-Adapting Streamlined BFT Protocols with Trusted Execution Environments}},
year = {2024},
booktitle = {IEEE IPDPS}
}

@inproceedings{TrustedHardwareAssisted,
author = {Zhao, Liangrong and Decouchant, Jérémie and Liu, Joseph and Lu, Qinghua and Yu, Jiangshan},
year = {2024},
title = {{Trusted Hardware-Assisted Leaderless Byzantine Fault Tolerance Consensus}},
booktitle = {IEEE TDSC}
}

@inproceedings{FlexiTrust,
      title={{Dissecting BFT Consensus: In Trusted Components we Trust!}},
      author={Suyash Gupta and Sajjad Rahnama and Shubham Pandey and Natacha Crooks and Mohammad Sadoghi},
      year={2023},
      booktitle={EuroSys}
}

@inproceedings{ENGRAFT,
author = {Wang, Weili and Deng, Sen and Niu, Jianyu and Reiter, Michael K. and Zhang, Yinqian},
title = {{ENGRAFT: Enclave-guarded Raft on Byzantine Faulty Nodes}},
year = {2022},
booktitle = {ACM CCS},
}

@INPROCEEDINGS{TBFT,
  author={Zhang, Jiashuo and Gao, Jianbo and Wang, Ke and Wu, Zhenhao and Li, Yue and Guan, Zhi and Chen, Zhong},
  title={{TBFT: Efficient {Byzantine} Fault Tolerance Using Trusted Execution Environment}}, 
  year={2022},
  booktitle={IEEE ICC}
}

@inproceedings{bluefish,
      author = {Victor Shoup},
      title = {Blue fish, red fish, live fish, dead fish},
      booktitle = {Cryptology {ePrint} Archive, Paper 2024/1235},
      year = {2024},
}

@inproceedings{thresholdlatency,
      title={The Latency Price of Threshold Cryptosystem in Blockchains}, 
      author={Zhuolun Xiang and Sourav Das and Zekun Li and Zhoujun Ma and Alexander Spiegelman},
      year={2024},
      booktitle={arXiv preprint arXiv:2407.12172},
}

@inproceedings{brachaRBC,
author = {Bracha, Gabriel},
title = {{Asynchronous Byzantine Agreement Protocols}},
year = {1987},
booktitle = {Inf. Comput.}
}

@INPROCEEDINGS{SgxPectre,
  author={Chen, Guoxing and Chen, Sanchuan and Xiao, Yuan and Zhang, Yinqian and Lin, Zhiqiang and Lai, Ten H.},
  booktitle={IEEE EuroS\&P}, 
  title={{SgxPectre Attacks: Stealing Intel Secrets from SGX Enclaves via Speculative Execution}}, 
  year={2019}
}

@misc{intel_software_security_guidance,
  title        = {{Intel. Software Security Guidance}},
  howpublished = {\url{https://www.intel.com/content/www/us/en/developer/topic-technology/software-security-guidance/overview.html}},
  note         = {(Accessed: 04-29-2026)}
}

@misc{sgxsimulate,
  title = {{How to Run Intel® Software Guard Extensions' Simulation Mode}},
  howpublished={\url{https://www.intel.com/content/www/us/en/developer/articles/training/usage-of-simulation-mode-in-sgx-enhanced-application.html}},
  note = {(Accessed: 04-29-2026)}
}

@inproceedings{jovanovic2024mahi,
  title={Mahi-Mahi: Low-Latency Asynchronous {BFT} {DAG}-Based Consensus},
  author={Jovanovic, Philipp and Kogias, Lefteris Kokoris and Kumara, Bryan and Sonnino, Alberto and Tennage, Pasindu and Zablotchi, Igor},
  booktitle={arXiv preprint arXiv:2410.08670},
  year={2024}
}

@inproceedings{consensusinpartialsync,
author = {Dwork, Cynthia and Lynch, Nancy and Stockmeyer, Larry},
title = {{Consensus in the Presence of Partial Synchrony}},
year = {1988},
booktitle = {JACM}
}

@inproceedings{Chandra1996,
author = {Chandra, Tushar Deepak and Toueg, Sam},
title = {{Unreliable Failure Detectors for Reliable Distributed Systems}},
year = {1996},
booktitle = {JACM}
}

@inproceedings{Clement2009,
author = {Clement, Allen and Wong, Edmund and Alvisi, Lorenzo and Dahlin, Mike and Marchetti, Mirco},
title = {{Making Byzantine Fault Tolerant Systems Tolerate Byzantine Faults}},
year = {2009},
booktitle = {USENIX NSDI}
}

@inproceedings {Nimble,
author = {Sebastian Angel and Aditya Basu and Weidong Cui and Trent Jaeger and Stella Lau and Srinath Setty and Sudheesh Singanamalla},
title = {{Nimble: Rollback Protection for Confidential Cloud Services}},
booktitle = {USENIX OSDI},
year = {2023},
}

@inproceedings{CacheBleed,
      author = {Yuval Yarom and Daniel Genkin and Nadia Heninger},
      title = {{CacheBleed: A Timing Attack on OpenSSL Constant Time RSA}},
      year = {2016},
      booktitle = {CHES}
}

@inproceedings{foreshadow,
  author    = {Van Bulck, Jo and Minkin, Marina and Weisse, Ofir and Genkin, Daniel and Kasikci, Baris and Piessens, Frank and Silberstein, Mark and Wenisch, Thomas F. and Yarom, Yuval and Strackx, Raoul},
  title     = {{Foreshadow: Extracting the Keys to the Intel SGX Kingdom with Transient Out-of-Order Execution}},
  booktitle = {USENIX Security},
  year      = {2018}
}

@inproceedings{TeeRollup,
  author={Wen, Xiaoqing and Feng, Quanbi and Lyu, Hanzheng and Niu, Jianyu and Zhang, Yinqian and Feng, Chen},
  booktitle={IEEE Transactions on Computers}, 
  title={TeeRollup: Efficient Rollup Design Using Heterogeneous TEE}, 
  year={2025}
}

@inproceedings{cacheout,
  author    = {Van Schaik, Stephan and Minkin, Marina and Kwong, Andrew and Genkin, Daniel and Yarom, Yuval},
  title     = {{CacheOut: Leaking Data on Intel CPUs via Cache Evictions}},
  booktitle = {IEEE S\&P},
  year      = {2021}
}

@inproceedings{crosstalk,
  author    = {Ragab, Hany and Milburn, Alyssa and Razavi, Kaveh and Bos, Herbert and Giuffrida, Cristiano},
  title     = {{CrossTalk: Speculative Data Leaks Across Cores Are Real}},
  booktitle = {IEEE S\&P},
  year      = {2021}
}

@inproceedings{aepic-leak,
  author    = {Borrello, Pietro and Kogler, Andreas and Schwarzl, Martin and Lipp, Moritz and Gruss, Daniel and Schwarz, Michael},
  title     = {{AEPIC Leak: Architecturally Leaking Uninitialized Data from the Microarchitecture}},
  booktitle = {USENIX Security},
  year      = {2022}
}

@misc{sgaxe,
  author    = {van Schaik, Stephan and Kwong, Andrew and Genkin, Daniel and Yarom, Yuval},
  title     = {{SGAxe: How SGX Fails in Practice}},
  year      = {2020},
  howpublished = {\url{https://sgaxe.com}},
  note         = {(Accessed: 04-29-2026)}
}

@inproceedings{plundervolt,
  author    = {Murdock, Kit and Oswald, David and Garcia, Flavio D. and Van Bulck, Jo and Gruss, Daniel and Piessens, Frank},
  title     = {{Plundervolt: Software-based Fault Injection Attacks against Intel SGX}},
  booktitle = {IEEE S\&P},
  year      = {2020}
}

@inproceedings{Shamirsecret,
author = {Shamir, Adi},
title = {{How to Share a Secret}},
year = {1979},
booktitle = {CACM},
}

@inproceedings{resdbdemopaper,
author = {Rahnama, Sajjad and Gupta, Suyash and Qadah, Thamir M. and Hellings, Jelle and Sadoghi, Mohammad},
title = {{Scalable, Resilient, and Configurable Permissioned Blockchain Fabric}},
year = {2020},
booktitle = {PVLDB},
}

@inproceedings{resdbpaper,
author = {Gupta, Suyash and Rahnama, Sajjad and Hellings, Jelle and Sadoghi, Mohammad},
title = {ResilientDB: global scale resilient blockchain fabric},
year = {2020},
booktitle = {VLDB},
}

@inproceedings{2023fever,
      title={{Fever: Optimal Responsive View Synchronisation}}, 
      author={Andrew Lewis-Pye and Ittai Abraham},
      year={2023},
      booktitle={arXiv preprint arXiv:2301.09881}
}

@inproceedings{honeybadger,
author = {Miller, Andrew and Xia, Yu and Croman, Kyle and Shi, Elaine and Song, Dawn},
title = {{The Honey Badger of BFT Protocols}},
year = {2016},
booktitle = {ACM CCS},
}

@inproceedings{asyncnetwork,
author = {Chevrou, Florent and Hurault, Aur\'{e}lie and Qu\'{e}innec, Philippe},
title = {{On the diversity of asynchronous communication}},
year = {2016},
Booktitle = {FAC}
}

@book{schrijver_theory,
author = {Schrijver, Alexander},
title = {Theory of linear and integer programming},
year = {1986},
isbn = {0471908541},
address = {USA}
}

@misc{aptos,
  title        = {{Aptos}},
  howpublished = {\url{https://aptosnetwork.com/}},
  note         = {(Accessed: 04-29-2026)}
}

@misc{sui,
  title        = {{Sui}},
  howpublished = {\url{https://www.sui.io/}},
  note         = {(Accessed: 04-29-2026)}
}

@misc{osmosis,
  title        = {{Osmosis: The premier decentralized crypto exchange}},
  howpublished = {\url{https://osmosis.zone/}},
  note         = {(Accessed: 04-29-2026)}
}

@misc{jupiter,
  title        = {{Jupiter | The DeFi Superapp}},
  howpublished = {\url{https://jup.ag/}},
  note         = {(Accessed: 04-29-2026)}
}

@misc{fides_appendix,
  title        = {{Fides Appendix}},
  howpublished = {\url{https://github.com/apache/incubator-resilientdb/blob/fides/fides_appendix.pdf}},
  note         = {(Accessed: 04-29-2026)}
}

@inproceedings{hotstuff-2,
      author = {Dahlia Malkhi and Kartik Nayak},
      title = {{HotStuff}-2: Optimal Two-Phase Responsive {BFT}},
      booktitle = {Cryptology {ePrint} Archive, Paper 2023/397},
      year = {2023}
}

@inproceedings{hotstuff-1,
author = {Kang, Dakai and Gupta, Suyash and Malkhi, Dahlia and Sadoghi, Mohammad},
title = {HotStuff-1: Linear Consensus with One-Phase Speculation},
year = {2025},
booktitle = {ACM SIGMOD}
}

@inproceedings{carry,
  author =	{Gupta, Suyash and Kang, Dakai and Malkhi, Dahlia and Sadoghi, Mohammad},
  title =	{{Brief Announcement: Carry the Tail in Consensus Protocols}},
  booktitle =	{DISC},
  year =	{2025}
}

@inproceedings{spotless,
  author={Kang, Dakai and Rahnama, Sajjad and Hellings, Jelle and Sadoghi, Mohammad},
  booktitle={IEEE ICDE}, 
  title={SpotLess: Concurrent Rotational Consensus Made Practical Through Rapid View Synchronization}, 
  year={2024}
}

@inproceedings{hydra,
      title={HYDRA: Breaking the Global Ordering Barrier in Multi-BFT Consensus}, 
      author={Hanzheng Lyu and Shaokang Xie and Jianyu Niu and Mohammad Sadoghi and Yinqian Zhang and Cong Wang and Ivan Beschastnikh and Chen Feng},
      year={2025},
      booktitle={arXiv preprint arXiv:2511.05843}
}

@misc{aws-nitro-enclaves,
  author       = {{Amazon Web Services}},
  title        = {{AWS Nitro Enclaves}},
  howpublished = {\url{https://aws.amazon.com/ec2/nitro/nitro-enclaves/}},
  note         = {(Accessed: 04-29-2026)}
}

@misc{azure-confidential-computing,
  author       = {{Microsoft Azure}},
  title        = {{Azure Confidential Computing}},
  howpublished = {\url{https://learn.microsoft.com/en-us/azure/confidential-computing/}},
  note         = {(Accessed: 04-29-2026)}
}

@misc{alibaba-confidential-computing,
  author       = {{Alibaba Cloud}},
  title        = {{Security features for trusted and confidential computing}},
  howpublished = {\url{https://www.alibabacloud.com/help/en/ecs/user-guide/overview-of-security-capability}},
  note         = {(Accessed: 04-29-2026)}
}

\clearpage
\appendix
\section{Building Blocks} \label{app:build_block}
\subsection{Reliable Broadcast (RBC)} \label{app:rbc} 
RBC is a fundamental primitive in distributed systems to ensure the consistent delivery of messages in the presence of Byzantine faults. 
In RBC, a sender replica $p_i$ first broadcasts a message $m$ in a round $r\in\mathbb{N}$ by invoking $\textsc{Bcast}(m, r)$. By running the protocol, a replica $p_k$ may output $Deliver (m, r, i)$, where $m$ is the broadcast message, $r$ is the round number, and $i$ is the identifier of the sender replica. RBC guarantees the following properties: 
\begin{packeditemize}
    \item \textbf{Validity:} If a correct replica $p_j$ invokes $\textsc{Bcast}(m, r)$, then every correct replica $p_i$ eventually invokes $Deliver(m, r, j)$.

    \item \textbf{Integrity:} A replica delivers a message $m$ at most once, and only if $m$ was broadcast by replica $p_j$ via $\textsc{Bcast}(m, r)$.

    \item \textbf{Agreement:} If a correct replica $p_i$ invokes $Deliver(m, r, j)$, then all correct replicas eventually invoke $Deliver(m, r, j)$.
\end{packeditemize}

Asynchronous DAG-based protocols rely on RBC to deliver blocks, accounting for 62.3\% of the latency, as shown in \secref{sec:motivation}.
Thus, we propose an efficient \textit{TEE-assisted Reliable Broadcast (T-RBC)} with linear message complexity, one step of message broadcast (instead of two), and tolerance of a minority of Byzantine nodes (\secref{subsec:mic}).

\subsubsection{T-RBC Protocol Flow.} \label{app:subsec:trbc-flow}
\figref{algorithm:teerbc} gives the detailed pseudocode and message pattern of T-RBC in \sysname. The primary routine \textsc{TBCast} runs on every broadcast and, in the good case, incurs only \textbf{a single broadcast}. The auxiliary routine \textsc{TCatchup} incurs two additional message rounds and is triggered only when a replica observes a missing vertex in the causal history.

\begin{figure}[h!]
    \vspace{-2mm}
    \begin{myprotocol}
    \setcounter{ALC@line}{0}
    \SPACE \textbf{\underline{TEE-assisted Reliable Broadcast (T-RBC) at replica $p_i$}}

    \FUNCTION{$\textsc{TBCast}$}{$v$}
        \STATE \textbf{Step 0: Monotonic Counter Value Generation}
        \STATE\quad $v.\textit{counter} \leftarrow \textsc{MC.GetCounter}(\textit{round-cert}, v)$
        \SPACE\quad \COMMENT{TEE generates a unique signed counter value}

        \STATE \textbf{Step 1: Broadcast}
        \STATE\quad As a sender, broadcast $\langle \textsc{TBCast} ,v \rangle$.
        \STATE\quad As a receiver, deliver $v$ upon receiving $\langle \textsc{TBCast} ,v \rangle$ with a valid counter value.
    \ENDFUNCTION

    \SPACE

    \FUNCTION{$\textsc{TCatchup}$}{$v', p_i$}
        \STATE \textbf{Step 2: Request for the missing vertex}
        \STATE\quad \textbf{Upon observing} a missing vertex $v'$ in the causal history of $v$ from $p_j$,
        replica $p_i$ actively sends $\langle \textsc{Request}, v' \rangle$ to $p_j$.

        \STATE \textbf{Step 3: Send the vertex $v'$.}
        \STATE\quad \textbf{Upon receiving} a valid $\langle \textsc{Request}, v' \rangle$ from a replica $p_i$, replica $p_j$ will then send out the full vertex $\langle \textsc{TCatchup} ,v' \rangle$ which $p_j$ requested.
        \STATE\quad $p_i$ \textbf{delivers} $v'$ upon receiving a valid $\langle \textsc{TCatchup} ,v' \rangle$.
    \ENDFUNCTION

    \SPACE

    \begin{tikzpicture}[>=stealth, thick, scale=0.95]
        \draw (0,0) -- (2.7,0);
        \node[left] at (0,0) {R1};
        \draw (0,-0.6) -- (2.7,-0.6);
        \node[left] at (0,-0.6) {R2};
        \draw (0,-1.2) -- (2.7,-1.2);
        \node[left] at (0,-1.2) {R3};

        \draw[dotted] (2.7,0) -- (3.6,0);
        \draw[dotted] (2.7,-0.6) -- (3.6,-0.6);
        \draw[dotted] (2.7,-1.2) -- (3.6,-1.2);

        \draw (3.6,0) -- (6.8,0);
        \draw (3.6,-0.6) -- (6.8,-0.6);
        \draw (3.6,-1.2) -- (6.8,-1.2);

        \draw[<->, blue, thick] (0.3,0.3) -- (1.2,0.3) node[midway, above] {\small Step 0};
        \draw[<->, blue, thick] (1.2,0.3) -- (2.7,0.3) node[midway, above] {\small Step 1};
        \draw[<->, blue, thick] (2.7,0.3) -- (3.6,0.3) node[midway, above] {\small ...};
        \draw[<->, blue, thick] (3.6,0.3) -- (5.1,0.3) node[midway, above] {\small Step 2};
        \draw[<->, blue, thick] (5.1,0.3) -- (6.6,0.3) node[midway, above] {\small Step 3};

        \node at (0.75, 0) {\includegraphics[width=0.55cm]{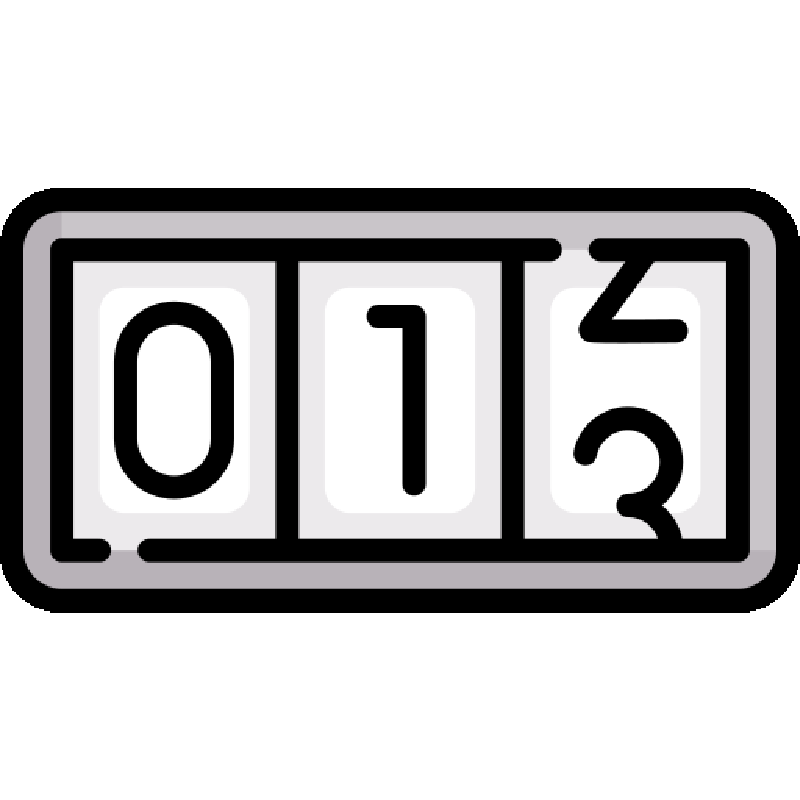}};
        \node at (1.05, -0.17) {\includegraphics[width=0.25cm]{icons/shield.png}};

        \draw[->] (1.2,0) -- (2.7,-0.6);
        \draw[->, red, dashed] (1.2,0) -- (2.7,-1.2);

        \coordinate (M) at ($(1.2,0)!0.6!(2.7,-1.2)$);
        \node[red, font=\bfseries] (fail) at (1.1,-0.9) {Failed};
        \draw[->, red] (fail) -- (M);

        \draw[->] (3.6,-1.2) -- (5.1,-0.6);

        \draw[->] (5.1,-0.6) -- (6.6,-1.2);
    \end{tikzpicture}

    \end{myprotocol}

    \vspace{-2mm}
    \caption{TEE-assisted Reliable Broadcast (\textit{T-RBC}) with \textit{MC}.}
    \vspace{-5mm}
    \label{algorithm:teerbc}
\end{figure}

\subsection{Common Coin} \label{app:commoncoin}
The common coin (also known as the global random coin) enables random decision-making among a set of participants, playing an essential role in asynchronous consensus~\cite{honeybadger,DAGRider,Dumbo-MVBA}. 
In \sysname, the common coin is utilized in every wave $w \in \mathbb{N}$. A replica invokes the coin by calling \Name{coinLeader}$(w)$ to select $p_i$ as the leader of wave $w$.
The common coin provides the following guarantees:

\begin{packeditemize}
\item \textbf{Agreement:} If two correct replicas invoke $\Name{coinLeader}(w)$ to obtain $p_i$ and $p_k$, respectively, then $p_i = p_k$.

\item \textbf{Termination:} Every correct replica invoking the function will eventually receive a result when the function has been invoked by more than $f$ replicas. 

\item \textbf{Unpredictability:} If no more than $f$ replicas have invoked \Name{coinLeader}$(w)$, the outcome is computationally indistinguishable from a random value, except with negligible probability $\epsilon$.
\end{packeditemize}

Traditional common-coin implementations \cite{thresholdlatency} that use resource-intensive cryptographic techniques are expensive, as indicated in \secref{sec:motivation}. 
Thus, we leverage TEEs to realize a more efficient \textit{TEE-assisted Common Coin (T-Coin)} (\secref{subsec:rang}). 

\section{Correctness Analysis of Building Blocks}
\subsection{TEE-assisted Reliable Broadcast (T-RBC)} \label{app:proof:tee_rbc}

\begin{theoremapx}
    The TEE-assisted Reliable Broadcast in \sysname satisfies \textbf{Validity}.
\end{theoremapx}
\begin{proof}
Suppose a correct initiator $p_j$ invokes \textsc{TBCast}$(v)$ at round $r$.
By calling the Monotonic Counter \textit{MC} inside the TEE, $p_j$ obtains a unique TEE-signed counter for $(p_j,r)$, ensuring that at most one valid message can be produced for that round.
Any replica delivers $v$ upon receiving a $\langle\textsc{TBCast},v\rangle$ accompanied by a valid \textit{MC} certificate.

If a Byzantine adversary selectively delivers the $\langle\textsc{TBCast},v\rangle$ message to only a subset of correct replicas, any such receiver will reference $v$ in its next-round vertex, making $v$’s metadata visible system-wide.
By the causality property of DAG, other correct replicas will eventually observe $v$ and thus trigger \textsc{TCatchup} to fetch the certified payload and then deliver $v$.
Hence, if a correct initiator invokes $\textsc{TBCast}(v)$, all correct replicas eventually $Deliver(v)$, satisfying \textit{Validity}.
\end{proof}

\begin{theoremapx}
    The TEE-assisted Reliable Broadcast in \sysname satisfies \textbf{Integrity}.
\end{theoremapx}
\begin{proof}
By design, \textit{MC} enforces at-most-one valid counter per sender and round: only the authorized initiator $v.source$ can produce a unique valid counter for round $v.round$.
If an adversary attempts to forge the vertex $v$ or create multiple conflicting messages for the same round, such messages will fail the TEE validation. 
Because no two messages share the same counter and replicas reject uncertified messages, each vertex is delivered at most once and only if genuinely broadcast by $v.source$. 
This guarantees \textit{Integrity}.
\end{proof}

\begin{theoremapx}
    The TEE-assisted Reliable Broadcast in \sysname satisfies \textbf{Agreement}.
\end{theoremapx}
\begin{proof}
Assume some correct replica $p_i$ has $Deliver_i(v)$.
Either $p_i$ received a valid $\langle\textsc{TBCast},v\rangle$ directly, or it obtained $v$ via \textsc{TCatchup} after seeing $v$’s metadata in DAG references.
In both cases, $v$’s metadata is propagated in subsequent vertices, so every correct replica will eventually observe it and invoke \textsc{TCatchup} to retrieve the certified payload.
By Integrity, there is at most one valid payload per $(\textsf{source},\textsf{round})$, so all correct replicas deliver the same $v$.
Thus \textbf{Agreement} holds.
\end{proof}

\subsection{TEE-assisted Common Coin (T-Coin)} \label{app:proof:tee_coin}
\begin{theoremapx}
    The \textit{TEE-assisted Common Coin} in \sysname satisfies \textbf{Agreement}.
\end{theoremapx}
\begin{proof}
    When at least $f{+}1$ replicas invoke $\Name{coinLeader}(w)$ for the same wave $w$, they rely on the same shared seed and input parameters within their \textit{RNG} components. Since the \textit{RNG} operates deterministically based on these inputs and the shared cryptographic seed, it produces the same output at all correct replicas. Thus, the selected leader is identical across all correct replicas, satisfying the \textbf{Agreement} property. 
\end{proof}

\begin{theoremapx}
    The \textit{TEE-assisted Common Coin} in \sysname satisfies \textbf{Termination}.
\end{theoremapx}
\begin{proof}
    If at least $f{+}1$ replicas invoke $\Name{coinLeader}(w)$, sufficient valid \textit{MC} values are generated to construct the required quorum proof. The reliable broadcast mechanism ensures that these \textit{MC} values, embedded in the corresponding vertices, are eventually delivered to all correct replicas. Consequently, any correct replica invoking $\Name{coinLeader}(w)$ can collect the necessary quorum proof and obtain a response from the \textit{RNG} component, thus ensuring \textbf{Termination}.
\end{proof}

\begin{theoremapx}
    The \textit{TEE-assisted Common Coin} in \sysname satisfies \textbf{Unpredictability}.
\end{theoremapx}
\begin{proof}
    As long as fewer than $f{+}1$ replicas have invoked \Name{coinLeader}$(w)$, the quorum proof required by the \textit{RNG} component cannot be formed. Without this proof, the output of the \textit{RNG} remains inaccessible and unpredictable to any adversary. Additionally, the internal state of the \textit{RNG} is securely protected within the TEE, preventing external entities from predicting its output. Therefore, the outcome remains computationally indistinguishable from random, satisfying the \textbf{Unpredictability} property.
\end{proof}

\section{Correctness Analysis for Reliable Dissemination} \label{app:reliable_diss}
\begin{claim} \label{claim:v_agreement}
    If a correct replica $p_i$ adds a vertex $v$ to its DAG, then eventually, all correct replicas add $v$ to their DAGs.
\end{claim}

\begin{proof}
We prove the claim by induction on the rounds.

\iheading{Base Case}: In the initial round, all correct replicas share the same initial state, so the claim holds trivially.

\iheading{Inductive Step}: Assume that for all rounds up to \( r - 1 \), any vertex added by a correct replica is eventually added to the DAGs of all other correct replicas. 

Now consider a vertex \( v \) added by a correct replica \( p_i \) in round \( r \).  When \( v \) is reliably delivered to $p_i$ (line \ref{line:r_deliver}), $p_i$ conducts validation, including verifying $v$'s validity (line \ref{line:vertex_validity}) and ensuring all referenced vertices by $v$ are already in the $DAG$ (line \ref{line:vertices_in_DAG}).

By the \textit{Agreement} property of reliable broadcast (\secref{app:rbc}), \( v \) will eventually be delivered to all correct replicas.
Since \( v \) passed these checks at \( p_i \), and all correct replicas have consistent DAGs for earlier rounds by the induction hypothesis, \( v \) will pass validation at all correct replicas.
Thus, all correct replicas will eventually add \( v \) to their DAGs.
\end{proof}

\begin{lemma} \label{lm:DAGview}
All correct replicas will eventually have an identical DAG view.
\end{lemma}

\begin{proof}
From Claim~\ref{claim:v_agreement}, any vertex added by a correct replica will eventually be added to the DAGs of all other correct replicas.
Despite network asynchrony and variations in message delivery times, Claim \ref{claim:v_agreement} guarantees that all correct replicas receive the same set of messages. Since all correct replicas apply the same deterministic rules to incorporate these vertices into their DAGs, the DAGs of all correct replicas will converge to the same state over time.
Thus, all correct replicas will eventually have identical DAG views.
\end{proof}

\section{Correctness Analysis for \sysname} \label{app:correctproof}
\subsection{Safety Proof} \label{app:subsec:safety}
\begin{lemma}\label{clm:uniquewaveleader}
When electing the leader vertex of wave $w$, for any two correct replicas $p_i$ and $p_j$, if $p_i$ elects $v_i$ and $p_j$ elects $v_j$, then $v_i = v_j$.
\end{lemma}
\begin{proof}

To prove that $v_i=v_j$, we need to ensure that all correct replicas use the same inputs when calling the random function \textit{RNG} (e.g., the vertices from the corresponding round in line \ref{line:rand}) to elect the leader of the wave $w$. From Lemma \ref{lm:DAGview}, we know that all correct replicas will eventually have identical DAGs. 
When a correct replica attempts to elect the leader for wave $w$, it does so after it has advanced to the necessary round and has incorporated all required vertices from previous rounds. Since the DAGs are identical and the random seed in \textit{RNG} is the same across replicas, the inputs to the \textit{Rand} function will be identical. Therefore, the output $p_j$ selected as the leader will be the same at all correct replicas.  Consequently, the vertex $v$ proposed by $p_j$ in round $w$ will be the same at all correct replicas, so $v_i=v_j$.
\end{proof}

\begin{lemma}\label{lm:leadervertexpath}
    If the leader vertex $v$ of wave $w$ is committed by a correct replica $p_i$, then for any leader vertex $v'$ committed by any correct replica $p_j$ in a future wave $w'>w$, there exists a path from $v'$ to $v$.
\end{lemma}

\begin{proof}
A leader vertex $v$ in wave $w$ is committed only if at least $f{+}1$ third-round vertices in wave $w$ reference $v$. Since each vertex in the DAG must reference at least $f{+}1$ vertices from the previous round, quorum intersection ensures that all vertices in the first round of wave $w+1$ have paths to $v$. By induction, this property extends to all vertices in all rounds of waves $w'>w$. Thus, every leader vertex $v'$ in waves $w'>w$ has a path to $v$, proving the lemma.
\end{proof}

\begin{lemma}\label{lm:ascendingcommit}
    A correct replica $p_i$ commits leader vertices in ascending order of waves.
\end{lemma}

\begin{proof}
In the algorithm, when a leader vertex $v$ of wave $w$ is committed (either directly or indirectly), the replica pushes $v$ onto a stack (line \ref{line:pushleader1} or line \ref{line:pushleader2}). The algorithm then recursively checks for earlier uncommitted leader vertices in waves $w'$ from $w-1$ down to the last decided wave (line \ref{line:recursivecheck}).
If such a leader $v'$ exists and there is a path from $v$ to $v'$, $v'$ is also pushed onto the stack. Since the stack is a LIFO (Last-In-First-Out) structure when popping vertices to commit (Line \ref{line:popleader}), the leaders are committed in order from the earliest wave to the current wave $w$. This ensures that leader vertices are committed in ascending order of waves.
\end{proof}

\begin{lemma}\label{lm:sameorder}
All correct replicas commit leader vertices in the same ascending order of waves.
\end{lemma}

\begin{proof}
We assume that a correct replica $p_i$ commits leader vertex $v$ of wave $w$ and another correct replica $p_j$ commits leader vertex $v'$ of wave $w'$, where $w < w'$.  By Lemma~\ref{lm:leadervertexpath}, there must exist a path from $v'$ to $v$ in the DAG.  This path implies that $v$ should have been indirectly committed by $p_j$. Thus, if a leader vertex is committed by one correct replica, then no other replica can skip committing it. 

Combining this observation with Lemma~\ref{lm:ascendingcommit}, all correct replicas commit leader vertices in the same ascending order of waves. 
\end{proof}

\begin{theorem}\label{thm:safety}
    \sysname guarantees the \textbf{Safety} property.
\end{theorem}

\begin{proof}

The commit procedure, as defined in the function \textsc{OrderVertices} (lines \ref{line:ordervertices_begin}--\ref{line:ordervertices_end}), is responsible for committing vertices.

By Lemma~\ref{lm:sameorder}, all correct replicas commit leader vertices in the same ascending order of waves. Furthermore, each committed leader's causal history is delivered in a predefined, deterministic order (line \ref{line:predefined_order}). This ensures that the sequence of vertices---and, by extension, the sequence of transactions contained within these vertices---is identical for all correct replicas.

Additionally, Lemma \ref{lm:DAGview} establishes that all correct replicas maintain identical DAGs. Since all correct replicas apply the same deterministic commit rules, the assignment of sequence numbers $sn$ to transactions is consistent across all replicas. Consequently, if two correct replicas $p_i$ and $p_j$ commit transactions $tx_1$ and $tx_2$ with sequence number $sn$, it follows that $tx_1=tx_2$. 
Thus, \sysname satisfies the \textbf{Safety} property.
\end{proof}

\subsection{Liveness Proof} \label{app:subsec:liveness}

\subsubsection{Preliminaries (aligned with Sec.~\ref{sec:core}).}
We adopt the $k$-iteration common-core abstraction of Sec.~\ref{sec:core}.
Let $N{=}2f{+}1$ replicas with up to $f$ Byzantine.
In a wave of $k$ rounds, each replica $i$ outputs a set $U_i$ after $k$ aggregation rounds, and the \emph{common core} is
\[
S^\ast(k)\;:=\;\bigcap_{i\in \{1,\dots,N\}} U_i .
\]
At the end of the wave, \sysname invokes a \textit{T-Coin} to elect a leader from the first round.
Let $\mathcal{L}$ denote the leader domain. 
We consider a \emph{global} coin that elects uniformly from all $N$ IDs (replicas).
A wave commits if the elected leader lies in $S^\ast(k)$.
Hence, the per-wave success probability is
\[
p_k\;=\;\frac{|S^\ast(k)|}{|\mathcal{L}|}\qquad(|\mathcal{L}|=N).
\]
Therefore the expected commitment latency (in \emph{rounds}) is
\begin{equation}\label{eq:exp-lat}
\mathbb{E}[\textnormal{rounds}]\;=\;k\cdot\frac{1}{p_k}\;=\;k\cdot\frac{|\mathcal{L}|}{|S^\ast(k)|}.
\end{equation}

We keep the DAG-side notation as follows:
round sets $R_1,\dots,R_\ell$, parent sets $P_{t+1}(\cdot)$ of size $f{+}1$, and
for $x\in R_1$ the carrier sets $U_x^{(t)}$ and final support
$supp_\ell(x):=|\{z\in R_\ell:\ P_\ell(z)\cap U_x^{(\ell-2)}\neq\varnothing\}|$.

\subsubsection{DAG $\Rightarrow$ Common-Core (bridge).} \label{app:bridge}
Set $k=\ell$. 
If a first-round vertex $x\in R_1$ is supported by at least $f{+}1$ final-round ($l$-th) vertices in the DAG, i.e., $supp_\ell(x)\ge f{+}1$, then $x$ appears in at least $f{+}1$ replicas' $k$-round aggregate $U_i$; hence 
\[
x\in S^\ast(k)\ :=\ \bigcap_{i\in \{1,\dots,N\}} U_i.
\]
Therefore, any DAG-side lower bound on 
$\#\{x\in R_1:\ supp_\ell(x)\ge f{+}1\}$ 
lifts verbatim to a lower bound on $|S^\ast(k)|$, and the per-wave success probability with a global random coin is 
$p_k=|S^\ast(k)|/N$.
In what follows, we work entirely in the DAG view (via $supp_\ell(\cdot)$ and the constructions) to present counting arguments and adversarial schedules; by the bridge above, these results directly yield the required common-core properties and liveness bounds.

\subsubsection{Proof by DAG.}
Before proceeding with the proof, we summarize the key symbols used throughout.

\begin{itemize}
    \item $\ell$: the number of rounds in a single wave.
    \item $R_i$: the disjoint vertex sets in round-$i$ of the DAG. By DAG construction rules, $R_i\in [f{+}1,2f{+}1], for\ i\in\{1,\dots,\ell\}$.
    \item $n_i$: the number of vertices in round $i$ ($n_i = |R_i|, for\ i\in\{1,\dots,\ell\}$).
    \item $P_{t+1}(w)\subseteq R_t$: The parent (referenced) set of $w\in R_{t+1}$, with $|P_{t+1}(w)|=f{+}1$ (all parents are distinct).
    \item $U_x^{(t)}$: the round-$t$ carrier set of $x\in R_1$, \ie, the set of round-$t$ vertices that support $x$.
    Define recursively as:
          \[
          \begin{aligned}
            \quad \quad U^{(2)}_x:&=\{u\in R_2:\ x\in P_2(u)\} \\
            \quad \quad U^{(t)}_x:&=\{w\in R_{t}: P_{t}(w)\cap U^{(t-1)}_x\neq\varnothing\}\ \ (t > 2).
          \end{aligned}
          \]
    \item $supp_{\ell}(x):$ The number of final-round ($l$-th) vertices that support (link to) $x\in R_1$.
          For $\ell$ rounds, $supp_\ell(x):=|\,\{z\in R_\ell:\ P_\ell(z)\cap U^{(\ell-1)}_x\neq\varnothing\}\,|$.
    \item \textit{Heavy / Light} vertex at round 1: a vertex $x\in R_1$ is \emph{heavy} if $|U_x^{(2)}|\ge n_2{-}f$
          (equivalently, $supp(x)=n_l$), and \textit{light} otherwise ($|U_x^{(2)}|\le m{-}f{-}1$).
          We write $H:=\{x\in R_1:\ |U_x^{(2)}|\ge n_2-f\}$ and $L:=R_1\setminus H$.
\end{itemize}

\bheading{Adversarial Network Model (Worst Case).}
We first prove the adversarial cases.

\begin{lemma}\label{lm:two_round_commit}
    In a \textbf{two-round} wave ($\ell=2$) with $N{=}2f{+}1$ replicas, \textbf{no} first-round vertex is guaranteed to reach $supp_2(x)\ge f{+}1$ whenever $f\ge2$ under an asynchronous adversary.
\end{lemma}

\begin{proof}
Assume the asynchronous adversary exists. 
In each round, consider any set $S$ of $f{+}1$ vertices in round 2 with $n_1\in [f+1,2f+1]$. 
Hence, the total number of links from $S$ to the first-round vertices is $(f+1)\cdot (f+1)$.

However, if we suppose each first-round vertex receives at most $f$ incoming links from $S$, the aggregate capacity would be $n_1 \cdot f$.
The difference, $(f+1)^2 - n_1 \cdot f$, evaluates to $-f^2+f+1$ in the worst-case scenario ($n_1 = 2f+1$).
Since this value is non-positive for all $f \ge 2$, the adversary can strategically construct a DAG where every vertex in round 1 has an indegree from $S$ of at most $f$. 
Consequently, under adversarial scheduling, the protocol may be forced to stall in every wave, as no vertex $x \in R_1$ can satisfy the required support threshold $supp_2(x) \ge f+1$.
\end{proof}

\begin{lemma} \label{lm:two_round_commit_exp}
    When employing a \textbf{two-round} commit strategy, liveness cannot be guaranteed. 
\end{lemma}

\begin{proof}

By Lemma~\ref{lm:two_round_commit}, under an asynchronous adversary and $f\ge2$, every $x\in R_1$ satisfies $supp_2(x)\le f$, so no vertex reaches the threshold $f{+}1$.
By the DAG$\Rightarrow$Common-Core bridge, this implies $S^\ast(2)=\varnothing$.

From the common-core preliminaries, a wave commits iff the coin-elected leader lies in $S^\ast(k)$; hence, for $k=2$ the per-wave success probability is
$p_2=\frac{|S^\ast(2)|}{|\mathcal L|}=0$ (Eq.~\eqref{eq:exp-lat}), yielding unbounded expected commit latency.
An asynchronous adversary can realize this in every wave, so progress cannot be guaranteed.
Therefore, liveness does not hold under a \textbf{two-round} commit strategy.
\end{proof}

\begin{lemma} \label{lm:heavy_suppf+1}
    In any wave with $\ell \ge 3$, every \textbf{heavy} vertex $x\in R_1$ satisfies $supp_\ell(x)=n_\ell\ge f{+}1$.
\end{lemma}

\begin{proof}
For $x\in R_1$ put $d_x:=|U^{(2)}_x|$, \ie, the number of round-2 vertices linked to $x$. 
Call $x$ \emph{heavy} if $d_x\ge n_2-f$; 
then $|R_2\setminus U_x^{(2)}|\le f$, so every $(f{+}1)$-subset of $R_2$ intersects $U^{(2)}_x$.

Consider any $y\in R_3$. 
Since $|P_3(y)|=f{+}1$ and $P_3(y)\subseteq R_2$, we have
$P_3(y)\cap U^{(2)}_x\neq\varnothing$, hence $y\in U^{(3)}_x$.
Therefore $U^{(3)}_x=R_3$ and $supp_3(x)=n_3\ge f{+}1$.

For any $t\ge3$, if $U^{(t)}_x=R_t$, then for any $z\in R_{t+1}$ we have $P_{t+1}(z)\subseteq R_t$, so trivially $P_{t+1}(z)\cap U^{(t)}_x\neq\varnothing$ and
$z\in U^{(t)}_x$. 
By induction, $U^{(t+1)}_x=R_{t+1}$ for all $t=3,\dots,\ell$.
Thus $supp_\ell(x)=|R_\ell|=n_\ell\ge f{+}1$.
\end{proof}

\begin{lemma} \label{thm:3layer-worst-all-in-one}
In any \textbf{three-round} wave semantics ($\ell=3$), at least two vertices $x\in R_1$ are \emph{heavy} and satisfy $supp(x)\ge f{+}1$. 
Moreover, this lower bound is attained when $n_1=n_2=2f{+}1$ (for any $n_3\in[f{+}1,2f{+}1]$).
\end{lemma}

\begin{proof}
For $x\in R_1$, put $d_x:=|\{u\in R_2:\ x\in P_2(u)\}|$ and call $x$
\emph{heavy} if $d_x\ge n_2-f$. By Lemma~\ref{lm:heavy_suppf+1}, any \emph{heavy} $x$ satisfies
$supp_3(x)=n_3\ge f{+}1$.

Let $H$ be the set of \emph{heavy} vertices.
Since each $u\in R_2$ has $f{+}1$ parents,
\[
\sum_{x\in R_1} d_x=\sum_{u\in R_2}|P_2(u)|=n_2(f{+}1).
\]
With $d_x\le n_2$ for \emph{heavy} vertices and $d_x\le n_2-f-1$ for \textbf{light},
\[
n_2(f{+}1)\ \le\ |H|\cdot n_2+(n_1-|H|)\cdot(n_2-f-1)=n_1(n_2-f-1)+|H|(f{+}1),
\]
thus we have
\[
|H|\ \ge\ n_1+n_2-\frac{n_1\,n_2}{f+1}.
\]

For fixed $n_2 \in[f{+}1,2f{+}1]$, the derivative of $n_1$ equals $1-\tfrac{n_2}{f+1}<0$, hence the right-hand side is minimized at the largest $n_1$, \ie, $n_1=2f{+}1$.
Plugging $n_1=2f{+}1$ yields
\[
|H|\ \ge\ (2f{+}1)-\frac{f\,n_2}{f+1},
\]
which is strictly decreasing in $n_2$, so the minimum over $n_2\in[f{+}1,2f{+}1]$ is attained
at $n_2=2f{+}1$:
\[
|H|\ \ge\ (2f{+}1)-\frac{f(2f{+}1)}{f+1}
\;=\;2-\frac{1}{f+1}.
\]
Taking ceilings gives $|H|\ge 2$ for all admissible $(n_1,n_2)$, and the \emph{minimum} is
attained at $(n_1,n_2)=(2f{+}1,2f{+}1)$ (independently of $n_3$).
\end{proof}

\begin{lemma}\label{lm:three_round_commit}
    In the \textbf{three-round} wave semantics ($\ell=3$), exactly two vertices in the first round are guaranteed to reach $supp_3(x)\ge f{+}1$ under an asynchronous adversary.
\end{lemma}

\begin{proof}

\bheading{Lower bound.}
From Lemma~\ref{thm:3layer-worst-all-in-one}, we get the lower bound of $|H|\ge 2$ for any set of $(n_1,n_2,n_3)$. 
Therefore, in every legal construction, at least two $x\in R_1$ reach $supp(x)=f{+}1$.

\smallskip
\bheading{Tightness via explicit construction.}
We next exhibit a legal construction with exactly two vertices reaching the threshold.

\paragraph{Step 1 (choose $P_2$).}
Pick two distinct \emph{heavy} vertices $h_1,h_2\in R_1$ and enforce
\[
\{h_1,h_2\}\ \subseteq\ P_2(u),\quad \forall\,u\in R_2.
\]
Thus $|U_{h_1}^{(2)}|=|U_{h_2}^{(2)}|=|R_2|=2f+1\ge f{+}1$.
Let $L:=R_1\setminus\{h_1,h_2\}$ be the set of the remaining $2f-1$ \textbf{Light} vertices.
For each $u\in R_2$ there remain exactly $f{-}1$ light-parent slots to fill so that $|P_2(u)|=f{+}1$. 
We will fill these using only a prescribed family of $f{+}1$ \emph{windows} on $R_2$ of size $f$ each, so that every light $x\in L$ is assigned to a window and all its descendants lie inside that window.
Formally, label $R_2$ cyclically as $\{0,1,\dots,2f\}$ and define
\[
S_i:=\{\,i,i{+}1,\dots,i{+}f{-}1\,\}\pmod{2f{+}1},
\quad i\in\{0,2,4,\dots,2f\}
\]
($f{+}1$ windows, each of size $f$).
Choosing the window in this way ensures that every column $u\in R_2$ is contained in either \(\lfloor f/2\rfloor\) or \(\lceil f/2\rceil\) of the windows \(S_i\), yielding near-uniform coverage.
This makes it easy to set the fractional column sums to $f$ or $f{-}1$ and lift to an integral solution by total unimodularity later.  

Further, we partition $L$ into $f{+}1$ groups $G_i$ so that $|G_i|=2$ for $i=0,2,\dots,2f-4$ (that is, for $f{-}2$ windows), $|G_i|=1$ for the remaining three windows; this uses all $2f-1$ light vertices.
Assign each light $x\in G_i$ to a target degree $d_x\in\{f,f{-}1\}$ as follows: 
for every window $G_i$ with two light vertices, both lights get $d_x=f$; among the three windows with only one light vertex, two of them get $d_x=f$ and one gets $d_x=f{-}1$.
Then
\[
\sum_{x\in L} d_x
=(f{-}2)\cdot 2f+2\cdot f+(f{-}1)
=(2f{+}1)(f{-}1),
\]
which equals the total number of light descendants required across $R_2$.

Now we have a bipartite $b$-matching system, and we define variables $a_{x,u}\in\{0,1\}$ for $x\in L, u\in R_2$ with the constraints

\begin{equation}\label{eq:bmatching}
\left\{
\begin{aligned}
&\sum_{u\in R_2} a_{x,u} \;=\; d_x && (\forall\, x\in L),\\
&\sum_{x\in L} a_{x,u} \;=\; f-1 && (\forall\, u\in R_2),\\
&a_{x,u} \;=\; 0 && \text{if } u\notin S_{i(x)},\\
&a_{x,u} \;\in\;\{0,1\}
\end{aligned}
\right.
\end{equation}

where $i(x)$ is the index of the window that $x$ is assigned to.

By the previous rule, the vertices in $R_2$ linked to any $x \in R_1$ are designated by $S_{i(x)}$, thus consecutive.
The constraint matrix of \eqref{eq:bmatching} has the \textbf{\emph{consecutive-ones}} property on rows (each row has $1$’s only within the interval $S_{i(x)}$), hence it is \textbf{totally unimodular} (see~\cite{schrijver_theory}). 
Therefore, any feasible fractional solution is integral; in particular, the system admits an integer solution whenever it admits a fractional one.
A fractional solution is immediate: for each $x\in L$ put
\(
a_{x,u}=d_x/f\ \text{for }u\in S_{i(x)}
\)
and $a_{x,u}=0$ otherwise. 
Then row-sums are $d_x$ by construction.
For the column-sums, by suitably permuting which windows are the two-light vs.\ one-light window, one achieves that for every $u$ the total 
\(\sum_{x: u\in S_{i(x)}} d_x = f\,(f{-}1)\),
hence the fractional column-sum equals
\(\sum_x a_{x,u}=(1/f)\cdot f(f{-}1)=f{-}1\).
By \emph{total unimodularity}, there is an \emph{integral} solution to \eqref{eq:bmatching}. 
Taking $P_2(u)$ to be $\{h_1,h_2\}\cup 
\{x\in L:\ a_{x,u}=1\}$ for each $u$ completes Step~1.
Thus for every light $x$ we have $|U_x^{(2)}|=d_x\le f$ and, moreover, \(U_x^{(2)}\subseteq S_{i(x)}\) is confined to a window of size $f$.

\paragraph{Step 2 (choose $P_3$).}
We first consider the worst but admissible case of having only $f{+}1$ vertices in $R_3$.
Use the $f{+}1$ vertices of $R_3$ to index the $f{+}1$ windows ($S_i$):
for each $i\in\{0,2,\dots,2f\}$ pick a distinct $v_i\in R_3$ and set
\[
P_3(v_i)\ :=\ R_2\setminus S_i.
\qquad (\text{\textbf{Note}: }|P_3(v_i)|=2f{+}1-f=f{+}1),
\]

\bheading{Verification.}
For the two \emph{heavy} vertices $h_1,h_2$ we have $|U_{h_j}|=2f{+}1\ge f{+}1$, hence 
\[
supp_3(h_j)=|R_3|=f{+}1.
\]

For any light $x\in G_i$, we have $U_x^{(2)}\subseteq S_i$ by Step~1, hence \(P_3(v_i)\cap U_x^{(2)}=\varnothing\). 
Therefore $x$ is \emph{not} hit by some $v_i$ and
\[
supp_3(x)\ \le\ |R_3|-1\ =\ f.
\]
Thus exactly only the two vertices $h_1,h_2$ satisfy $supp_3(\cdot)=f{+}1$.

Therefore, combining the lower and tightness yields exactly two vertices in the first round that are guaranteed to reach $supp_3(\cdot) \ge f{+}1$, which completes the proof.
\end{proof}

\begin{lemma} \label{lm:three_round_commit_exp}
    When employing a \textbf{three-round} commit strategy ($\ell=3$), \sysname commits a vertex leader every \textbf{$3f+2$} rounds in the DAG under an asynchronous adversary in expectation.
\end{lemma}

\begin{proof}
Consider any wave $w$. By Lemma \ref{lm:three_round_commit}, at least \textbf{two} $x\in R_1$ satisfy $supp_2(x) \ge f{+}1$.
The common coin (Sec.~\ref{app:commoncoin}) ensures that the adversary cannot predict the coin's outcome before the leader is revealed. 
By the DAG$\Rightarrow$Common-Core bridge, this implies $|S^\ast(3)|=2$.
From the common-core preliminaries, a wave commits iff the coin-elected leader lies in $S^\ast(k)$; hence for $k=3$ the per-wave success probability is
$p_3=\frac{|S^\ast(3)|}{|\mathcal L|}=\frac{2}{2f+1}$.
By Eq.~\eqref{eq:exp-lat}, the committed latency is 
\[
\mathbb{E}[\textnormal{rounds}]\;=\;k\cdot\frac{1}{p_k}\;=3\cdot\frac{2f+1}{2}\;=3f+\frac{3}{2}.
\]
Therefore, in expectation, when employing a \textbf{three-round} commit strategy, \sysname commits a vertex leader every $\lceil 3f+\frac{3}{2} \rceil = 3f+2$ rounds in the DAG.    
\end{proof}

\begin{lemma} \label{lm:four_round_commit}
    When employing \textbf{four-round} wave semantics ($\ell=4$), at least $f{+}1$ vertices $x\in R_1$ are guaranteed to satisfy $supp_4(x)\ge f{+}1$.
\end{lemma}

\begin{proof}
For any $y\in R_3$ we define its \emph{gap} as $S(y):=R_2\setminus P_3(y)$, so $|S(y)|=f$.

\paragraph{Step 1 (a single protected class consumes $f{+}1$ third-layer vertices).}
Say that $y$ \emph{avoids} $x$ iff $U^{(2)}_x\subseteq S(y)$. 
For a fixed window
$S\subseteq R_2$ of size $f$, define the class
\[
\mathcal{C}_2(S):=\{x\in R_1:\ U^{(2)}_x\subseteq S\}.
\]
All $x\in\mathcal{C}_2(S)$ are avoided exactly by those $y$ with $S(y)=S$.

To make any $x$ \emph{eligible} to be avoided by some $z\in R_4$,
we must have $|R_3\setminus U^{(3)}_x|\ge f{+}1$, \ie
\begin{equation}\label{eq:avoid-cond}
A^{(3)}_x\ :=\ |\{y\in R_3:\ P_3(y)\subseteq R_2\setminus U^{(2)}_x\}|\ \ge\ f{+}1.
\end{equation}

Crucially, these $f{+}1$ avoiders cannot be combined across different gaps: 
to cover the entire class $\mathcal{C}_2(S)$, one needs $f{+}1$ vertices $y\in R_3$ all with the
\emph{same} gap $S$. 
Since $|R_3|\le 2f{+}1<2(f{+}1)$, at most one window $S^\star$ can receive $\ge f{+}1$ such vertices; hence at most one class $\mathcal{C}_2(S^\star)$ can be fully protected (\ie, every $x\in\mathcal{C}_2$ is possible to be avoided by (not linked to) vertex in $R_4$).

\paragraph{Step 2 (there are $\boldsymbol{f{+}1}$ vertices outside any single window).}
Fix any size-$f$ set $S\subseteq R_2$ and consider the columns $R_2\setminus S$.
Each $u\in R_2\setminus S$ must choose $f{+}1$ parents in $R_1$, so the total number of \emph{outside-$S$} parent edges equals $(n_2-f)(f{+}1)$. 
Every $x\in\mathcal{C}_2(S)$ has $U^{(2)}_x\subseteq S$ and thus contributes \emph{zero} edge outside $S$.
Thus, all these ``outside-$S$'' parent selections must be drawn from vertices outside $\mathcal{C}_2(S)$, so the number of distinct vertices outside $\mathcal{C}_2(S)$ satisfies:
\begin{equation}\label{eq:outside-count}
|\{x\in R_1:\ U^{(2)}_x\nsubseteq S\}|\ \ge\
\frac{(n_2-f)(f{+}1)}{n_2-f}\ =\ f{+}1.
\end{equation}
Therefore, there must be at least $f{+}1$ $R_1$ vertices outside $\mathcal{C}_2(S)$.

\paragraph{Step 3 (outside vertices cannot be protected and must be hit by round 4).}
Choose the third layer so that all its gaps equal one fixed window $S^\star$
(\ie\ $S(y)=S^\star$ for all $y\in R_3$). 
Then the class $\mathcal{C}_2(S^\star)$ has $A^{(3)}_x=|R_3|\ge f{+}1$, hence $\eqref{eq:avoid-cond}$ holds for all $x\in\mathcal{C}_2(S^\star)$. 
By Step~1, no other class can satisfy $\eqref{eq:avoid-cond}$.
In particular, every $x\notin\mathcal{C}_2(S^\star)$ has $A^{(3)}_x\le f$, so $|R_3\setminus U^{(3)}_x|\le f$ and \emph{no} vertex $z\in R_4$ can avoid $x$ (since $P_4(z)$ must have size $f{+}1$).
Consequently, all $z\in R_4$ hit every $x\notin\mathcal{C}_2(S^\star)$ and $supp_4(x)=|R_4|=n_4\ge f{+}1$.
Applying \eqref{eq:outside-count} to $S^\star$ guarantees at least $f{+}1$ such $x$, proving the lower bound.

\paragraph{Step 4 (tightness).}
Since we have at most $f$ Byzantine replicas, it is possible for them to make $n_1=f{+}1$.
Then, every $u\in R_2$ must choose all vertices of $R_1$, so each $x\in R_1$ is \emph{heavy} and by Lemma~\ref{lm:heavy_suppf+1} we have $supp_4(x)=n_4\ge f{+}1$. 
Thus, the lower bound is attained.

This proves the lemma: in the \textbf{four-round} wave semantics ($\ell=4$), at least $f{+}1$ vertices $x\in R_1$ are guaranteed to satisfy $supp_4(x)\ge f{+}1$.
\end{proof}

\begin{lemma}\label{lm:four_round_commit_exp}
When employing a \textbf{four-round} commit strategy ($\ell=4$), \sysname commits a vertex leader in at most \textbf{8} rounds in the DAG under an asynchronous adversary in expectation.
\end{lemma}

\begin{proof}
Consider any wave $w$. By Lemma~\ref{lm:four_round_commit}, at least $f{+}1$ first-round vertices satisfy $supp_4(x)\ge f{+}1$.
By the DAG$\Rightarrow$Common-Core bridge, this gives $|S^\ast(4)|\ge f{+}1$.
From the common-core preliminaries, a wave commits iff the coin-elected leader lies in $S^\ast(k)$; hence for $k=4$ the per-wave success probability is
$p_4=\frac{|S^\ast(4)|}{|\mathcal L|}\ \ge\ \frac{f{+}1}{2f{+}1}.$.
By Eq.~\eqref{eq:exp-lat}, the committed latency is 
\[
\mathbb{E}[\textnormal{rounds}]\;=\;k\cdot\frac{1}{p_k}\;=4\cdot\frac{2f+1}{f+1}\;\le 8.
\]
therefore, in expectation, when employing a \textbf{four-round} commit strategy, \sysname commits a vertex leader every \textbf{8} rounds in the DAG.  
\end{proof}

\begin{lemma} \label{lm:lge4-upper-fplus1}
Fix $N{=}2f{+}1$ in any admissible DAG. 
For any wave length $\ell\ge4$, in the worst case
\[
\Bigl|\{\,x\in R_1:\ \mathrm{supp}_{l}(x)\ge f{+}1\,\}\Bigr|\ =\ f{+}1.
\]
\end{lemma}

\begin{proof}
(\emph{Upper bound.})
In the worst case the adversary controls $f$ replicas and can suppress all of their round-1 proposals, so $n_1\in[f{+}1,\,2f{+}1]$ can be driven down to $|R_1|=f{+}1$. 
Consequently,
\[
\bigl|\{\,x\in R_1:\ \mathrm{supp}_{l}(x)\ge f{+}1\,\}\bigr|
\ \le\ n_1\ =\ f{+}1,
\]
so at most $f{+}1$ first-round vertices can reach the threshold (at round $\ell$) under this worst case.

\smallskip
(\emph{(Tightness at $\ell=4$ and persistence for $\ell>4$).})
By Lemma~\ref{lm:four_round_commit}, in any $\ell=4$ wave we \emph{always} have at least $f{+}1$ vertices $x\in R_1$ with $supp_4(x)\ge f{+}1$.
Combining with the upper bound above yields equality at $\ell=4$:
\[
\Bigl|\{\,x\in R_1:\ supp_{4}(x)\ge f{+}1\,\}\Bigr|\;=\;f{+}1
\quad\text{(worst case)}.
\]

Any wave semantics with more than four rounds ($\ell>4$) yields at least as many qualifying vertices per wave, i.e., $|S^\ast(\ell)|\ge f{+}1$.
However, an adversary can always force $n_1=f{+}1$ independently of $\ell$, so no larger worst-case guarantee is achievable for any $\ell>4$.
Therefore, the worst-case cap remains $f{+}1$ for all $\ell\ge4$.
\end{proof}

\begin{lemma}\label{lm:four-rounds-optimal}
Consider waves of fixed constant length $\ell$ in a DAG-BFT protocol with $n{=}2f{+}1$ replicas, where each round-$r\ge2$ vertex references exactly $f{+}1$ parents in round $r{-}1$.
Then $\ell=4$ is optimal among constant-round wave designs: it is the smallest constant that guarantees constant-round $8$ commitment, and any $\ell>4$ does not improve the worst-case guarantee.
\end{lemma}

\begin{proof}
By Lemmas~\ref{lm:two_round_commit_exp} and \ref{lm:three_round_commit_exp}, waves of length $2$ or $3$ can not guarantee constant-round commitment.
Further, by Lemma~\ref{lm:four_round_commit_exp}, waves of length $4$ \emph{do} guarantee constant-round commitment (in particular, $8$ rounds in expectation in the worst case).

Moreover, by Lemma~\ref{lm:lge4-upper-fplus1}, for any $\ell>4$ the adversary can still only force $|R_1|=f{+}1$, so the worst-case number of round-1 vertices that can reach the threshold is capped by $f{+}1$. 
Increasing $\ell$ therefore does not improve the worst-case guarantee but only adds commit latency.

Hence $\ell=4$ is the optimal constant length of wave that guarantees constant-round commitment under an $n{=}2f{+}1$ system.
\end{proof}

\bheading{Random-Delay Network Model (Stochastic Case).}
We next prove the liveness of \sysname within the random-delay network model.

\begin{lemma}[Commit probability under random delays]\label{lem:four-round-prob}
Let $N{=}2f{+}1=:n$ and consider a wave of length $\ell=4$ with $n_1=n_2=n_3=n_4=n$.
Assume each round-$r$ vertex independently chooses exactly $f{+}1$ parents
uniformly at random from $R_{r-1}$.
Then
\[
p_4(f)\ \ge\ 0.9419, \quad\text{for all } f\ge 1,
\]
with \(p_4(f)\) strictly increasing in \(f\) and \(p_4(f)\!\to\!1\) as \(f\!\to\!\infty\).
Consequently, the expected rounds-per-commit satisfies
\[
\mathbb{E}[\textnormal{rounds per commit}]
\;=\;\frac{4}{p_4(f)}
\;<\;4.25,
\]
and \(\mathbb{E}[\textnormal{rounds per commit}]\to 4\) rapidly as \(f\) grows.
\end{lemma}

\begin{proof}
By symmetry of the coin, fix the leader $v\in R_1$.  
A wave commits iff this $v$ satisfies the round-4 support threshold $\,supp_4(v)\ge f{+}1$,
which (by the DAG$\Rightarrow$common-core bridge) is equivalent to $v\in S^\ast(4)$.

\emph{Round 2.}
Each $u\in R_2$ includes $v$ among its $f{+}1$ parents with probability $p_1=\frac{f+1}{n}$.
Let $X:=\#\{u\in R_2:\ v\in P_2(u)\}\sim\mathrm{Bin}(n,p_1)$.

\emph{Round 3.}
Conditioned on $X=k$, a given $y\in R_3$ hits at least one of those $k$ ``good'' $R_2$ vertices with
\[
q_3(k)=1-\frac{\binom{n-k}{f+1}}{\binom{n}{f+1}}\;,
\]
hence we define \(T := |S_3(v)|\), the number of Round 3 vertices that link to $v$.
Conditional on \(X = k\), the random variable \(T\) follows a binomial distribution with parameters \((n, q_3(k))\); \ie
\[
T \mid X = k \sim \mathrm{Bin}\bigl(n,\, q_3(k)\bigr).
\]

\emph{Round 4.}
Conditioned on \(T=t\), there are \(t\) “good’’ vertices in \(R_3\) (those in \(S_3(v)\)).
A vertex \(z\in R_4\) chooses \(f{+}1\) parents uniformly from \(R_3\); it avoids \(S_3(v)\) with probability \(\binom{n-t}{f+1}/\binom{n}{f+1}\), hence it hits \(S_3(v)\) with
\[
q_4(t)\;=\;1-\frac{\binom{n-t}{f+1}}{\binom{n}{f+1}}\,.
\]
Since different \(z\in R_4\) choose independently, the number \(Y:=\#\{z\in R_4:\ z\rightsquigarrow v\}\) satisfies
\[
Y \,\big|\, T=t \;\sim\; \mathrm{Bin}\bigl(n,\,q_4(t)\bigr).
\]

\emph{Commit event and total probability.}
The four-round commit rule is $Y\ge f{+}1$. Taking the law of total probability over
$X$ and $T$ yields exactly the stated triple sum $p_4(f)=\Pr[Y\ge f{+}1]$.
In this random-delay model, the expected commit time is written as:
$\mathbb{E}[\textnormal{rounds}]=4/p_4(f)$.

For small \(f\), evaluating \(p_4(f)\) yields
\[
p_4(1)\approx 0.9419,\quad
p_4(2)\approx 0.9869,\quad
p_4(3)\approx 0.9970, \dots
\]
so the corresponding expected rounds per commit are
\(4/p_4(f)\approx 4.247,\,4.053,\,4.012\), respectively.
Moreover, \(p_4(f)\to 1\) rapidly as \(f\) grows, implying that one wave (four rounds)
suffices in expectation to commit.

This completes the proof.
\end{proof}

\iheading{How this ties back to common-core latency.}
By the DAG $\Rightarrow$ Common-Core correspondence (Sec.~\ref{sec:core}),
the event ``coin-elected leader $v$ satisfies $supp_4(v)\ge f{+}1$'' is exactly ``$v\in S^\ast(4)$''. 
Thus the per-wave commit probability is $p_4=\Pr[v\in S^\ast(4)]$, and the expected commit latency formula from Eq.~\eqref{eq:exp-lat} specializes to $\mathbb{E}[\textnormal{rounds}]=4/p_4$ under the random-delay model.

\bheading{System-Wide Liveness.}
\begin{lemma} \label{lm:vertex_commit}
    Every vertex $v$ in the DAG will eventually be committed by every correct replica.
\end{lemma}
\begin{proof}
    When a correct replica creates a new vertex, it establishes strong edges pointing to vertices from the previous round (line~\ref{strongedge}) and weak edges pointing to vertices with no alternative path in the DAG (line~\ref{weakedge}). 
    This construction ensures that all vertices are reachable and eventually included in the total order. 

    By Theorem \ref{thm:safety}, all correct replicas commit the leader vertices, along with their corresponding causal history, in the same order.
    Furthermore, Lemma~\ref{lm:four-rounds-optimal} and Lemma~\ref{lem:four-round-prob} demonstrate that \sysname guarantees vertices are committed within a constant number of rounds, ensuring they are eventually committed with probability 1.

    Thus, every correct replica will eventually commit $v$ within a constant number of rounds. 
\end{proof}

\begin{theorem} \label{thm:Liveness}
    \sysname satisfies the \textbf{Liveness} property.
\end{theorem}

\begin{proof}
In \sysname, assume a correct replica $p_i$ receives transaction $tx$. 
Upon receiving $tx$, $p_i$ creates a vertex $v$ encapsulating $tx$ along with potentially other transactions. 
The vertex $v$ is then reliably broadcast to all other replicas.

By Lemma~\ref{lm:vertex_commit}, every vertex $v$ in the DAG is eventually committed by every correct replica. 
Since the transaction $tx$ is contained in the vertex $v$, it will also be committed by all correct replicas. 
Thus, \sysname ensures that any transaction received by a correct replica will eventually be committed by all correct replicas, satisfying the \textbf{Liveness} property. 
\end{proof}

\section{SGX Hardware vs.\ Simulation Overhead}
\label{sec:sgx-overhead}

To justify our use of SGX simulation mode (SGX-SIM) in geo-distributed WAN experiments, we conduct a systematic evaluation of the performance overhead introduced by real SGX hardware encryption (SGX-HW) compared to SGX-SIM in two direct measurement settings: local (no network) and LAN. We then use these measurements to interpret the WAN-scale setting used in the main evaluation.

\subsection{Local Microbenchmark}

We first isolate the raw overhead of enclave transitions and Memory Encryption Engine (MEE) paging by measuring individual ECALL latencies without any network involvement. Table~\ref{tab:sgx-micro} summarizes the results across four representative ECALL types.

\begin{table}[h]
\captionsetup{labelfont={bf}}
\centering
\caption{SGX-HW vs.\ SGX-SIM ECALL microbenchmark.}
\label{tab:sgx-micro}
\small
\begin{tabular}{lrrr}
\toprule
\textbf{ECALL Type} & \textbf{SGX-SIM} & \textbf{SGX-HW} & \textbf{Overhead} \\
\midrule
Empty (transition only) & 7.62\,$\mu$s & 10.46\,$\mu$s & +2.84\,$\mu$s (+37.2\%) \\
Compute (4\,KB payload) & 10.68\,$\mu$s & 14.13\,$\mu$s & +3.45\,$\mu$s (+32.3\%) \\
Memory copy (4\,KB) & 7.95\,$\mu$s & 11.11\,$\mu$s & +3.16\,$\mu$s (+39.7\%) \\
Memory copy (64\,KB) & 17.64\,$\mu$s & 21.22\,$\mu$s & +3.58\,$\mu$s (+20.3\%) \\
\bottomrule
\end{tabular}
\end{table}

SGX-HW introduces a consistent overhead of approximately 3\,$\mu$s per ECALL due to hardware-enforced security validations, Enclave Page Cache (EPC) lookups, and MEE encryption/decryption. The relative overhead exceeds 30\% for small payloads but decreases for larger payloads as the computation itself begins to dominate.

\subsection{LAN Benchmark}

We next measure the end-to-end latency in a LAN setting, where a client on one machine issues requests to a server on another machine within the same datacenter. We use a 4\,KB payload to match the microbenchmark compute workload.

\begin{table}[h]
\captionsetup{labelfont={bf}}
\centering
\caption{End-to-end latency in LAN (same datacenter).}
\label{tab:sgx-lan}
\small
\begin{tabular}{lr}
\toprule
\textbf{Environment} & \textbf{End-to-End Latency ($\mu$s)} \\
\midrule
LAN + SGX-SIM & ${\sim}93$ \\
LAN + SGX-HW & ${\sim}93$ (within noise) \\
\bottomrule
\end{tabular}
\end{table}

Although a single ECALL adds ${\sim}3\,\mu$s in isolation, at LAN scale the enclave transition overlaps with concurrent socket writes and kernel scheduling rather than serializing on the critical path. The residual difference between SGX-HW and SGX-SIM falls within the run-to-run jitter ($\sim$1--5\,$\mu$s) from userspace measurements, making the two statistically indistinguishable in this setting.

\subsection{EPC Capacity and Enclave Memory Usage}

The SGX-enabled \texttt{ecs.g7t.xlarge} instances used in our experiments provide 3.875\,GiB of EPC. Each replica runs one enclave configured with an 8\,MiB heap and two 4\,MiB stacks, for a reserved budget of approximately 16\,MiB per replica.

In our implementation, the enclave hosts only the minimal trusted components (\textit{MC}, \textit{RAC}, and \textit{RNG}). Runtime profiling shows that the peak in-enclave heap usage is only 4096\,B during our experiments. Consensus-level data structures such as the DAG, message queues, and transaction pools remain outside the enclave. As a result, the enclave memory footprint remains small and stable across the evaluated workloads, and none of our tested configurations approach EPC pressure or incur EPC paging.

\subsection{Summary}

These results confirm that SGX-HW overhead is only measurable in isolated, network-free microbenchmarks. Even in the most optimistic LAN setting, the overhead is already negligible ($<$0.1\%). Under WAN-scale latencies (e.g., the 100\,ms delays used in our DIN evaluation), the SGX-HW overhead becomes even more insignificant. This justifies our use of SGX-SIM for geo-distributed WAN experiments where SGX-enabled instances are not available across all cloud regions.



\end{document}